\documentclass{ws-igtr}

\usepackage{tikz,diagbox,relsize,amsmath,amssymb,bbm,algpseudocode,multicol,tabularx,enumerate,multicol}

\newtheorem{observation}{Observation}
\newtheorem{assumption}{Assumption}
\newtheorem{problem}{Problem}
\newtheorem{algorithm}{Algorithm}

\newcommand{\enquote}[1]{#1}
\providecommand{\natexlab}[1]{#1}
\providecommand{\url}[1]{\texttt{#1}}

\expandafter\ifx\csname urlstyle\endcsname\relax
  \else
  \fi

\newcommand{\R}{\mathbb{R}}

\DeclareMathOperator*{\argmin}{arg\,min}

\begin{document}

\markboth{Joe Clanin and Sourabh Bhattacharya}
{Additive Security Games: Structure and Optimization}

\catchline{}{}{}{}{}
Additive Security Games: Structure and Optimization
\title{ADDITIVE SECURITY GAMES: STRUCTURE AND OPTIMIZATION}

\author{JOE CLANIN}

\address{Department of Computer Science, Iowa State University\\ 226 Atanasoff Hall, 2434 Osborn Drive,\\ 
Ames, Iowa 50011, USA\\\email{jsc@iastate.edu}}

\author{SOURABH BHATTACHARYA}

\address{Department of Mechanical Engineering, Iowa State University\\
Department of Computer Science, Iowa State University\\
Ames, Iowa 50011, USA\\
sbhattac@iastate.edu
}

\maketitle

\begin{history}
\received{(Day Month Year)}
\revised{(Day Month Year)}
\end{history}

\begin{abstract}
	In this work, we provide a structural characterization of the possible Nash equilibria in the well-studied class of security games with additive utility. Our analysis yields a classification of possible equilibria into seven types and we provide closed-form feasibility conditions for each type as well as closed-form expressions for the expected outcomes to the players at equilibrium. We provide uniqueness and multiplicity results for each type and utilize our structural approach to propose a novel algorithm to compute equilibria of each type when they exist. We then consider the special cases of security games with fully protective resources and zero-sum games. Under the assumption that the defender can perturb the payoffs to the attacker, we study the problem of optimizing the defender expected outcome at equilibrium. We show that this problem is weakly NP-hard in the case of Stackelberg equilibria and multiple attacker resources and present a pseudopolynomial time procedure to solve this problem for the case of Nash equilibria under mild assumptions. Finally, to address non-additive security games, we propose a notion of {\it nearest additive game} and demonstrate the existence and uniqueness of a such a nearest additive game for any non-additive game.
\end{abstract}


\ccode{Subject Classification: 91A68.}

\section{Introduction}

The allocation of limited resources by interacting agents has long been a fundamental object of study in game theory. Classical examples include Colonel Blotto games (\cite{gross1950continuous,hausken2012impossibility,bhurjee2016existence}), Gale's {\it games of finite resources} (\cite{ferguson2000games}) such as inspection games or goofspiel, and fair division problems (\cite{brams1996fair}). Two-player non-cooperative resource allocation games have been the subject of extensive study. In such games, adversaries such as auditors and auditees in the case of audit games (\cite{blocki2013audit,blocki2015audit}), attackers and defenders in security domains (\cite{alpcan2010network,manshaei2013game,hausken2017information,hausken2021governments}), or environmental regulators and polluting agents (\cite{perera2022stackelberg}) compete to allocate their limited resources over some collection of objects. The latter class of games, known as security games, feature a defender with some limited number of defensive resources to protect a set of targets and a resource-constrained attacker seeking to attack some collection of these targets. Such games are the subject of the present work.

A common assumption in the literature of security games is that targets exhibit an independence property known as additivity. Informally, players receive payoffs only for targets that are attacked and the total payoff to a player for a given set of attacked targets is the sum of their payoffs for each target. Non-additive security games have received significantly less attention than additive games. The polynomial time computability of Nash and Stackelberg equilibria in Non-additive games are addressed by \cite{xu2016mysteries,wang2017non,wang2017security} in which it is shown that equilibrium computation in general non-additive security games is NP-hard and that the polynomial time solvable class of such games can be characterized in terms of the combinatorial problem encoded by the defender pure strategy space. Intuitively, the assumption of additivity allows for compact representation of security games leading to efficient solution algorithms. The existing literature on additive security games is stratified by several key factors:
\begin{itemize}
    \item Whether or not the game is zero-sum
    \item The number of resources available to the attacker: single versus multiple
    \item The manner of play: Nash equilibria in simultaneous games and Stackelberg equilibria in leader-follower games
    \item Size and structure of schedules: some resources may be able to cover multiple targets
    \item The homogeneity or heterogeneity of the player resources: whether or not all resources are applicable to all schedules
    \item Resource protection: whether or not resources are {\it fully protective} in the sense that no payoff is received when an attacked target is also covered by the defender
\end{itemize}

The standard multiple LP approach (\cite{conitzer2006computing}) for computing Stackelberg equilibria requires time which is polynomial in the number of player pure strategies. Thus, in security games where a player has multiple resources and pure strategies correspond to choosing a particular subset of targets, alternative solution methods must be devised to render the problem tractable. The first efficient algorithm to compute a Stackelberg equilibrium in security games with a single attacker resource was proposed in  \cite{kiekintveld2009computing}. This procedure, entitled `Efficient Randomized Allocation of Security Resources' or ERASER, employs a mixed-integer linear programming approach to address the case of singleton schedules and homogeneous resources. Then, after eliminating the assumption that assigning more than one defensive resource does not benefit the defender or harm the attacker any more than a single resource (an assumption made by the present work and by all other works discussed herein), \cite{kiekintveld2009computing} also proposes an algorithm, ORIGAMI, to compute a Stackelberg equilibrium and verifies its efficacy through simulation. The relationship between Stackelberg equilibria and Nash equilibria in security games is studied in \cite{korzhyk2011stackelberg}. This work shows that Stackelberg equilibria are also Nash equilibria in the case of a single attacker resource and under the assumption that any subset of a schedule is also a schedule but that this need not be the case when the attacker is allowed more than a single resource. A game is said to have interchangeable Nash equilibria if any profile in which both players are playing {\it some} equilibrium strategy is itself a Nash equilibrium. A further result of \cite{korzhyk2011stackelberg} is that Nash equilibria in additive security games are interchangeable, and thus there is no equilibrium selection problem in such games. Equilibria in games with multiple attacker resources were first addressed by \cite{korzhyk2011stackelberg}, but the authors propose no efficient solution algorithm and provide only experimental results for the sake of comparison of the properties of games with varying player resource counts and manners of play. The first polynomial time algorithm to compute Nash equilibria in games with multiple attacker resources is given in \cite{korzhyk2011security}. This procedure begins by solving a game in which the defender has no resources, gradually increases the resources available to the defender, and transitions between a number of phases in which the player best responses are computed until the final number of defender resources is reached. While effective, this quadratic time procedure and its description do little to elucidate the underlying structure of Nash equilibria in additive security games.

All approaches to the efficient computation of equilibria (both Nash and Stackelberg when tractable) in additive security games rely upon the compact representation of player strategies in terms of marginal probabilities rather than explicit mixed strategies represented as distributions over exponentially large pure strategy spaces. This raises the question of whether or not the computed marginal probabilities can actually be implemented by some pure strategy. This question is completely resolved for Stackelberg equilibria by \cite{korzhyk2010complexity} wherein the authors demonstrate that in the case of homogeneous resources and schedules of size at most 2 or in the case of heterogeneous resources and singleton schedules, such a pure strategy can be computed in polynomial time. The problem is shown to be NP-hard in all other cases, even for zero-sum games. Furthermore, existing approaches to equilibrium computation are shown to have the capability to produce solutions in terms of marginal probability which are not implementable by any mixed strategy when schedules exceed size 2. In light of this complexity result for games with arbitrary schedules and owing to the fact that such schedules are essential elements of real-world security game implementations such as those utilized by the U.S. Federal Air Marshals (FAMS) (\cite{tambe2011security}), a branch-and-bound approach to equilibrium computation in games with arbitrary schedules and a single attacker resource is proposed in \cite{jain2010security}. This algorithm, ASPEN (Accelerated SPARS Engine), generates branching rules and bounds using ORIGAMI and is shown to exhibit significant performance improvements over previously existing techniques.
Another approach to circumventing the hardness results given in \cite{korzhyk2010complexity} is studied in \cite{letchford2013solving} wherein a class of security games on graphs whose schedules satisfy a necessary condition for implementability given by the Bihierarchy Birkhoff-Von Neumann theorem (\cite{budish2013designing}). Positive results regarding polynomial time Stackelberg equilibrium computation are given for games with heterogeneous or homogeneous defender resources in which the schedules have the structure of paths in a collection of rooted trees or paths in a path graph (respectively).

A $O(n)$ procedure to compute a Stackelberg equilibrium in games with a single attacker resource, singleton homogeneous schedules, and resources that are fully protective in the sense that the attacker receives no payoff when an attacked target is defended (but the defender may incur some cost) is proposed by \cite{lerma2011linear}. It is also shown that this algorithm can be modified to compute Stackelberg equilibria in $O(n\log(n))$ time for with all of the same assumptions, but non-fully-protective resources. More recent work (\cite{emadi2019security,hamidgamesec20}) has shown that for zero-sum games with homogeneous singleton schedules, multiple attacker resources and fully protective resources, there exists a $O(n)$ procedure to compute Nash equilibria. In this work, we extend the structural approach adopted in \cite{emadi2019security,hamidgamesec20} to study Nash equilibria in general sum security games with multiple attacker resources.

In many security settings, it is reasonable to assume that a defender can alter the payoff structure of the game. Defenders may artificially increase the vulnerability or inflate the perceived importance of certain targets to address their true incentives regarding the task of attack mitigation. A natural optimization problem of maximizing the defender expected utility at equilibrium in a security game given the possibility of perturbation of the game parameters by the defender is addressed in \cite{shi2018designing}. The authors propose a mixed integer linear program approach with an approximation guarantee to solve this problem under a weighted $L^1$ norm constraint as well as a polynomial time approximation scheme for a restricted version of the problem. Furthermore, when attacker payoffs are restricted to closed bounded intervals, it is shown that the optimization problem can be solved in $O(n^2\log(n))$ time by utlizing a modified version of ORIGAMI. The approaches in \cite{shi2018designing}, however, only address the case of Stackelberg equilibria and a single attacker resource. As mentioned previously, by \cite{korzhyk2011stackelberg} this implies that the optimization problem the case of Nash equilibria in games with a single attacker resource can be solved in polynomial time but leaves open the case of Nash equilibria and multiple attacker resources. 

We make the following contributions to the theory of general-sum additive security games with singleton schedules and multiple homogeneous attacker and defender resources:
\begin{enumerate}
    \item By leveraging necessary structural properties of Nash equilibria, we give a characterization of the possible Nash equilibria into seven types and demonstrate the existence of a game exhibiting each type.
    \item We provide feasibility conditions for each type of equilibrium and propose a novel $O(m^3)$ algorithm to compute equilibria based upon our structural analysis which is far more intuitive than existing approaches.
    \item In the case of fully protective resources, we show that only five types of Nash equilibria are possible and that our equilibrium computation algorithm can be made to run in $O(m^2)$ time.
    \item For all types of equilibria, we derive closed-form expressions for the player equilibrium strategies and expected outcomes at equilibrium and characterize the uniqueness and multiplicity of the various types of equilibria.
    \item Under the assumption that the defender can perturb the payoffs to the attacker, we show that the problem of maximizing the expected outcome to the defender is weakly NP-hard in the case of Stackelberg equilibria and multiple attacker resources.
    \item We propose a pseudopolynomial time procedure based upon our theory of types to find a globally optimal solution to the problem of maximizing defender expected utility in the case of Nash equilibria under a disjointness assumption.
    \item We propose a notion of {\it nearest additive game} and demonstrate the existence and uniqueness of such a game for any (potentially non-additive) security game.
\end{enumerate}

\subsection{Security Game Formulation}
     We now formulate the security game model to be  considered in this work.
     Define a two player game between an attacker and defender over a target set $T=\{1,\dots,m\}$. The attacker will attack $1\leq k_a<m$ targets and the defender will defend $1\leq k_d<m$ targets. We assume that at most one resource is allocated to a target, an assumption introduced in \cite{korzhyk2011security}. If a target is defended, we say the target is `covered', and if a target is not defended we shall say this target is `uncovered'. Let $U_a^c:2^T\to (0,\infty)$ and $U_a^u:2^T\to(0,\infty)$ be functions giving the payoff to the attacker when a set of attacked targets is covered and uncovered, respectively. Similarly define $U_d^c,U_d^u:2^t\to(-\infty,0)$ giving the costs to the defender. We assume the payoffs satisfy
\begin{equation}\label{delta_conditions}
\Delta_a(S)=U_a^u(S)-U_a^c(S)>0\hspace{1em}\text{ and }\hspace{1em} \Delta_d(S)=U_d^c(S)-U_d^u(S)>0.
\end{equation}
That is, when a set of targets is attacked, it is better for the defender that the targets are defended and better for the attacker that the targets remain exposed. Let $s_1^a,\dots,s_{\binom{m}{k_a}}^a\subseteq T$ be the pure strategies of the attacker and let $s_1^d,\dots,s^d_{\binom{m}{k_d}}\subseteq T$ be the pure strategies of the defender. The payoffs to the attacker and defender for a given pure strategy profile $s_i^a,s_j^d$ are
\[U_a^c(s_i^a\cap s_j^d)+U_a^u(s_i^a\setminus s_j^d)\text{  and  }U_d^c(s_i^a\cap s_j^d)+U_d^u(s_i^a\setminus s_j^d)\]
respectively. Denote the payoff matrices to the attacker and defender by $A$ and $B$ respectively. By convention we assume the attacker is the row player. We say that a security game is {\it additive} if the payoff functions are additive over attacked targets. That is, $U_a^c(S)=\sum_{t\in S}U_a^c(\{t\})$, and the same is true for $U_a^u,U_d^c,U_d^u$. For convenience we will omit braces when payoff functions are evaluated at singletons.
As noted in the previous section, the case of additive payoff functions is a well-studied security game model (\cite{kiekintveld2009computing,korzhyk2010complexity,yin2010stackelberg,korzhyk2011stackelberg,korzhyk2011security,emadi2019security,hamidgamesec20,emadi2020characterization,emadi2021game}). All security games under consideration, unless stated otherwise, will be assumed to have additive payoff functions.

Under the assumption of additive payoffs, the expected outcome for each player at equilibrium depends only upon the marginal probabilities with which each target is attacked and defended. That is:
\begin{equation}\label{vavd}
\begin{aligned}
v_a&=p^TAq=\sum_{t=1}^m\left[\alpha_tU_a^c(t)\beta_t+\alpha_tU_a^u(t)(1-\beta_t)\right]\\
v_d&=p^TBq=\sum_{t=1}^m\left[\alpha_tU_d^c(t)\beta_t+\alpha_tU_d^u(t)(1-\beta_t)\right]\\
\end{aligned}
\end{equation}
Where $\alpha=[\alpha_1,\dots,\alpha_m]^T$ and $\beta=[\beta_1,\dots,\beta_m]^T$ are the attack and defense probability vectors (respectively) given by
\[\alpha_i=\sum_{i\in S^a_t}p_t\hspace{1em}\text{ and }\hspace{1em}\beta_i=\sum_{i\in S^d_t}q_t.\]
Note that $\alpha_i,\beta_i\in[0,1]$ (since $p,q$ are probability vectors) and 
\begin{equation}\label{sum_alpha_beta}
\sum_{i=1}^m\alpha_i=k_a\hspace{1em}\text{ and }\hspace{1em}\sum_{i=1}^m\beta_i=k_d
\end{equation}

It is worth noting that no information is lost by utilizing the compact representation given by \eqref{vavd}. It is noted in \cite{korzhyk2010complexity} that a mixed strategy $(p^*,q^*)$ which realizes any given $(\alpha^*,\beta^*)$ can be computed in $O(m^{4.5})$ time by utilizing the Birkhoff-Von Neumann theorem (\cite{horn2012matrix})  to decompose a doubly stochastic matrix into a convex combination of permutation matrices. In general such a $(p^*,q^*)$ need not be unique. We rephrase the definition of Nash equilibrium in terms of the compact representation \eqref{vavd} as follows: 

\begin{definition}
    A pair $(\alpha^*,\beta^*)$ is a Nash equilibrium if for all $\alpha,\beta$
    \[v_a(\alpha,\beta^*)\leq v_a(\alpha^*,\beta^*)\hspace{1em}\text{ and }\hspace{1em}v_d(\alpha^*,\beta)\leq v_d(\alpha^*,\beta^*).\]
\end{definition}
 
\section{Results}

We now present our structural approach to equilibrium computation in additive security games. 

\subsection{Structural Analysis of Nash Equilibria}
The following key observation is central to our analysis:
\begin{observation}\label{theobservation}
    Due to \eqref{vavd}, a player can unilaterally deviate from $(\alpha,\beta)$ if and only if they are able to shift marginal probability from one target to another. By definition, therefore, if $(\alpha^*,\beta^*)$ is a Nash equilibrium:
    \begin{equation*}
    \begin{aligned}
        \alpha_i^*>0,\alpha_j^*<1&\implies U_a^c(j)\beta_j+U_a^u(j)(1-\beta_j)\leq U_a^c(i)\beta_i+U_a^u(i)(1-\beta_i)\\
        \beta_i^*>0^,\beta_j^*<1&\implies \alpha_j\Delta_d(j)\leq\alpha_i\Delta_d(i)
    \end{aligned}
    \end{equation*}
    That is, no player can unilaterally deviate and increase their expected payoff.
\end{observation}
\begin{table}
        \centering
        \def\arraystretch{1.4}
		\begin{tabular}{| l | c | c | c | }
		\hline
		\backslashbox{$\mathlarger{\mathlarger{\beta_i^*}}$}{$\mathlarger{\mathlarger{\alpha_i^*}}$} & $=0$ & $\in(0,1)$ & $=1$ \\ 
		\hline 
		$=0$ & $i\in I_1$ & $i\in I_2$ & $i\in I_3$\\
		\hline
		$\in(0,1)$ &  $i\in I_4$ & $i\in I_5$& $i\in I_6$\\ 
		\hline
		$=1$ & $i\in I_7$& $i\in I_8$& $i\in I_9$\\
		\hline
		\end{tabular}	
		\caption{Definition of the sets $I_1,\dots,I_9$. For example, $I_8$ is the set $\{i:\alpha_i\in(0,1),\beta_i=1\}$.}\label{targetpartition}
		\label{partitiontable}
\end{table}
We make the following assumption regarding the game parameters:
\begin{assumption}\label{distinctassumption}
    The parameters $U_a^c(1),\dots,U_a^c(m),U_a^u(1),\dots,U_a^u(m)$ are distinct and the parameters $\Delta_d(1),\dots,\Delta_d(m)$ are distinct.
\end{assumption}
To make use of Observation~\ref{theobservation}, for any $(\alpha^*,\beta^*)$ we define a partition of the target set into nine sets $I_1,\dots,I_9$ defined by Table~\ref{targetpartition}.

\begin{lemma}\label{I4I7lemma}
    If $(\alpha^*,\beta^*)$ is a Nash equilibrium and $I_4\cup I_7\neq \emptyset$, then $I_2\cup I_3\cup I_5\cup I_6=\emptyset$.
\end{lemma}
\begin{proof}
    Suppose $i\in I_4\cup I_7$. Then $\alpha_i^*=0$ and $\beta_i^*\in (0,1]$. As $k_a\geq 1$ there exists $j$ so that $\alpha_j>0$. We must have $\beta_j=1$ for all $j$ such that $\alpha_j>0$ (else the defender has a feasible deviation increasing their payoff).
\end{proof}
This yields the following theorem:
\begin{theorem}\label{NEtheorem}
     A Nash equilibrium $(\alpha^*,\beta^*)$ is of precisely one of the following two types:
    \begin{itemize}
        \item Type I: $\forall i\in T$, $i\in T\setminus (I_4\cup I_7)$
        \item Type II: $\forall i\in T$, $i\in I_1\cup I_4\cup I_7\cup I_8\cup I_9$, $I_4\cup I_7\neq \emptyset$
    \end{itemize}
\end{theorem}
 By considering every pair of targets to which Observation~\ref{theobservation} applies, we arrive at the following lemmas which describe the necessary structural relationships between targets in the sets $I_1,\dots,I_9$ at equilibrium.
\begin{lemma}\label{c1c2lemma}
    If $(\alpha^*,\beta^*)$ is a Nash equilibrium, then
    \begin{enumerate}[(i)]
        \item There exists a constant $c_1$ so that $U_a^u(i)-\beta_i^*\Delta_a(i)=c_1$ for all $i\in I_2\cup I_5\cup I_8$ . Furthermore, for all $i\in I_2$ we have $c_1=U_a^u(i)$ and for all $i\in I_8$ we have $c_1=U_a^c(i)$.
        \item There exists a constant $c_2$ so that $\alpha_i^*\Delta_d(i)=c_2$ for all $i\in I_5\cup I_6$ . Furthermore, for $i\in I_6$ we have $c_2=\Delta_d(i)$.
    \end{enumerate}
\end{lemma}
\begin{proof}
    Assertion (i) follows from application of Observation~\ref{theobservation} to the sets $I_2,I_5$ and $I_8$. Similarly, assertion (ii) follows from application to $I_5$ and $I_6$.
\end{proof}
\begin{lemma}\label{defenderstructurallemma}
    If $(\alpha^*,\beta^*)$ is a Nash equilibrium, then
    \begin{enumerate}[(i)]
        \item For $i\in I_5\cup I_6\cup I_8\cup I_9$ and $j\in I_2\cup I_5$, $\alpha^*_j\Delta_d(j)\leq\Delta_d(i)$
        \item For $i\in I_5\cup I_6\cup I_8\cup I_9$ and $j\in I_3\cup I_6$, $\Delta_d(j)\leq\Delta_d(i)$. Furthermore $\Delta_d(j)<\Delta_d(i)$ for $j\in I_3,i\in I_6$.
        \item For $i\in I_5$ and $j\in I_6$, $\Delta_d(j)\leq\Delta_d(i)$
    \end{enumerate}
\end{lemma}

\begin{lemma}\label{attackerstructurallemma}
    If $(\alpha^*,\beta^*)$ is a Nash equilibrium, then
    \begin{enumerate}[(i)]
        \item For $i\in I_2\cup I_3\cup I_5\cup I_6\cup I_8\cup I_9$ and $j\in I_1\cup I_2$, $U_a^u(j)\leq U_a^u(i)$. Furthermore, $U_a^u(j)<U_a^u(i)$ for $j\in I_1$, $i\in I_2$.
        \item For $i\in I_9$ and $j\in I_2\cup I_5\cup I_8$, $U_a^c(j)\leq U_a^c(i)$. Furthermore, $U_a^c(j)<U_a^c(i)$ for $j\in I_5, i\in I_8$.
        \item For $i\in I_2\cup I_3\cup I_5\cup I_6\cup I_8\cup I_9$ and $j\in I_1\cup I_2\cup I_5\cup I_8$, $U_a^c(j)\leq U_a^u(i)$
    \end{enumerate}
\end{lemma}
In particular, Lemma~\ref{defenderstructurallemma} follows by considering all possible defender deviations and Lemma~\ref{attackerstructurallemma} follows by considering all possible attacker deviations.

\subsection{Computing Nash Equilibria}

We now utilize the structural characterization of Nash equilibria given in the section 2.1 to derive a novel algorithm to compute a Nash equilibrium, its type, and the expected outcomes to the players at equilibrium.

Note that by Assumption~\ref{distinctassumption} and Lemma~\ref{c1c2lemma} we have $|I_2|,|I_6|,|I_8|\in \{0,1\}$ and at most one of $I_2$ and $I_8$ is nonempty at equilibrium. 

Define six subtypes of Type I equilibrium as follows:
\begin{center}
    \begin{tabular}{|c|c|c|c|}
        \cline{2-4}
        \multicolumn{1}{c|}{} & {${I_2=\emptyset\atop I_8=\emptyset}$} & ${I_2\neq\emptyset\atop I_8=\emptyset}$&${I_2=\emptyset\atop I_8\neq\emptyset}$\\ \hline
        {$I_6=\emptyset$}&Type I.A.i&Type I.A.ii&Type I.A.iii\\ \hline
        {$I_6\neq\emptyset$}&Type I.B.i&Type I.B.ii&Type I.B.iii\\ \hline
    \end{tabular}
\end{center}

Introduce the parameters $r=|I_1|$, $s=|I_3|$ and $t=|I_9|$. The equilibrium computation algorithm will function as follows: For each possible value of the parameters $r,s,t$ and each $\texttt{type}\in\{$I.A.i,\dots,I.B.iii$\}$ we assign targets to the sets $I_1,\dots,I_9$ according to the necessary structural properties given by Lemmas~\ref{defenderstructurallemma} and~\ref{attackerstructurallemma}. Then, utilizing the constant sum property \eqref{sum_alpha_beta} of the marginal attack and defense probability vectors, Lemma~\ref{c1c2lemma}, and the definitions of the sets $I_1,\dots,I_9$ we calculate the necessary values of $\alpha_i,\beta_i$ for each $i$ as well as $c_1$ and $c_2$ in our candidate solution. The candidate equilibrium just constructed is then checked for feasibility.
We now give the details of the candidate equilibrium construction for each $\texttt{type}$ and a fixed $r,s,t$. When $I_2,I_6$ or $I_8$ is nonempty, denote their single elements by $j^{[2]}, j^{[6]}$ and $j^{[8]}$ respectively. Without loss of generality, assume the targets are ordered such that $U_a^u(i)<U_a^u(j)$ for $i<j$. The following procedure assigns the targets to the sets $I_1,\dots,I_9$:

\begin{enumerate}
    \item By Lemma~\ref{attackerstructurallemma} (i), for all types we take $I_1=\{1,\dots,r\}$.
    \item If $\texttt{type}\in\{$I.A.ii,I.B.ii$\}$ we set $I_2=\{r+1\}$ by Lemma~\ref{attackerstructurallemma} (i), otherwise $I_2=\emptyset$.
    \item By Lemma~\ref{defenderstructurallemma} (ii), $I_3$ is taken to be the $s$ targets of least $\Delta_d(i)$ in $T\setminus (I_1\cup I_2)$.
    \item If $\texttt{type}\in\{$I.B.i,I.B.ii,I.B.iii$\}$ we let $I_6$ contain the target of least $\Delta_d(i)$ in $T\setminus (I_1\cup I_2\cup I_3)$ by Lemma~\ref{defenderstructurallemma} (ii), otherwise $I_6=\emptyset$.
    \item By Lemma~\ref{attackerstructurallemma} (ii), we assign the $t$ targets of greatest $U_a^c(i)$ in $T\setminus (I_1\cup I_2\cup I_3\cup I_6)$ to $I_9$. 
    \item If $\texttt{type}\in\{$I.A.iii,I.B.iii$\}$, by Lemma~\ref{attackerstructurallemma} we let $I_8$ contain the target of greatest $U_a^c(i)$ in $T\setminus (I_1\cup I_2\cup I_3\cup I_6\cup I_9)$, otherwise $I_8=\emptyset$. 
    \item Finally, $I_5=T\setminus (I_1\cup I_2\cup I_3\cup I_6\cup I_8\cup I_9)$.
\end{enumerate}
\begin{table}
\caption{Expressions for $c_1$ and $c_2$ for each type of type I equilibrium and fixed values of $r,s,t$.}
    \setlength{\tabcolsep}{0.5em} 
    \renewcommand{\arraystretch}{2.7}
    \begin{multicols}{2}
    \begin{tabular}{|c|c|c|}
        \cline{2-3}
        \multicolumn{1}{c|}{} & $c_1$ & $c_2$\\ \hline
        I.A.i&${t-k_d+\sum_{j\in I_5}{U_a^u(j)\over \Delta_a(j)}\over \sum_{j\in I_5}{1\over \Delta_a(j)}}$&${k_a-s-t\over \sum_{j\in I_5}{1\over \Delta_d(j)}}$\\ \hline
        I.A.ii&$U_a^u(j^{[2]})$&${k_a-s-t-\alpha_{j^{[2]}}\over \sum_{j\in I_5}{1\over\Delta_d(i)}}$\\ \hline
        I.A.iii&$U_a^c(j^{[8]})$&${k_a-s-t-\alpha_{j^{[8]}}\over \sum_{j\in I_5}{1\over\Delta_d(i)}}$\\ \hline
    \end{tabular}
    \columnbreak
    \begin{tabular}{|c|c|c|}
        \cline{2-3}
        \multicolumn{1}{c|}{} & $c_1$ & $c_2$\\ \hline
        I.B.i&${\sum_{i\in I_5}{U_a^u(i)\over \Delta_a(i)}-k_d+t+\beta_{j^{[6]}}\over \sum_{i\in I_5}{1\over \Delta_a(i)}}$&$\Delta_d(j^{[6]})$\\ \hline
        I.B.ii&$U_a^u(j^{[2]})$&$\Delta_d(j^{[6]})$\\ \hline
        I.B.iii&$U_a^c(j^{[8]})$&$\Delta_d(j^{[6]})$\\ \hline
    \end{tabular}
    \end{multicols}
    \label{c1c2table}
\end{table}
After targets are assigned to $I_1,\dots,I_9$, we compute the constants $c_1$ and $c_2$ guaranteed by Lemma~\ref{c1c2lemma} using \eqref{sum_alpha_beta}. The result of this calculation is given in Table~\ref{c1c2table}. It remains to compute the necessary values of $\alpha_i$ and $\beta_i$ for each $i$. This is done as follows:
\begin{enumerate}
    \item $\alpha_i$ and $\beta_i$ for targets in $I_1,I_3$ and $I_9$ are assigned either $0$ or $1$ according to Table~\ref{partitiontable}.
    \item As required by Table~\ref{partitiontable}: If $I_2\neq \emptyset$, set $\beta_{j^{[2]}}=0$; If $I_6\neq \emptyset$, set $\alpha_{j^{[6]}}=1$; If $I_8\neq \emptyset$, set $\beta_{j^{[8]}}=1$.
    \item For each $i\in I_5$, by Lemma~\ref{c1c2lemma} we have
    \begin{equation}
    \alpha_i={c_2\over \Delta_d(i)}\hspace{1em}\text{and}\hspace{1em}\beta_i={U_a^u(i)-c_1\over \Delta_a(i)}
    \label{alphabetaI5}
    \end{equation}
    where $c_1$ and $c_2$ are given by the current value of $\texttt{type}$ and the expressions in Table~\ref{c1c2table}. If $\texttt{type}=$I.A.i, this completes the assignment of $\alpha$ and $\beta$.
    \item If $\texttt{type}=$I.B.ii, by \eqref{sum_alpha_beta} we have
    \begin{equation}
    \alpha_{j^{[2]}}=k_a-s-t-1-\sum_{j\in I_5}{\Delta_d(j^{[6]})\over\Delta_d(j)}\hspace{1em}\text{and}\hspace{1em}\beta_{j^{[6]}}=k_d-t-\sum_{j\in I_5}{U_a^u(j)-U_a^u(j_{[2]})\over\Delta_a(j)},
    \label{IBiialphabeta}
    \end{equation} completing the assignment of $\alpha$ and $\beta$.
    \item If $\texttt{type}=$I.B.iii, by \eqref{sum_alpha_beta} we have
    \begin{equation}
    \alpha_{j^{[8]}}=k_a-s-t-\sum_{j\in I_5}{\Delta_d(j^{[6]})\over \Delta_d(i)}\hspace{1em}\text{and}\hspace{1em}\beta_{j^{[6]}}=k_d-t-1-\sum_{j\in I_5}{U_a^u(j)-U_a^c(j^{[8]})\over \Delta_a(j)},
    \label{IBiiialphabeta}
    \end{equation} completing the assignment of $\alpha$ and $\beta$.
    \item If $\texttt{type}\in\{I.A.ii,I.A.iii,I.B.i\}$, the remaining single undetermined $\alpha_i$ or $\beta_i$ is calculated (or determined not to exist) during feasibility checking.
\end{enumerate}

Once the candidate equilibrium is constructed and $\alpha$ and $\beta$ are computed (with the possible exception of a single $\alpha_i$ or $\beta_i$ in the case of $\texttt{type}\in\{I.A.ii,I.A.iii,I.B.i\}$), we must verify whether or not the construction is a feasible Nash equilibrium. The following result, which we restate using our notation, is noted in \cite{korzhyk2011security}:
\begin{lemma}\label{c1c2feasibilitylemma}
    A pair of marginal attack and defense probability vectors $(\alpha,\beta)$ is a Nash equilibrium if and only if, for all $t\in T$
    \begin{equation*}
    \begin{aligned}
        \beta_t\neq0&\implies \alpha_t\Delta_d(t)\geq c_2\\
        \beta_t\neq 1&\implies \alpha_t\Delta_d(t)\leq c_2\\
        \alpha_t\neq 0&\implies \beta_tU_a^c(t)+(1-\beta_t)U_a^u(t)\geq c_1\\
        \alpha_t\neq 1&\implies \beta_tU_a^c(t)+(1-\beta_t)U_a^u(t)\leq c_1
    \end{aligned}
    \end{equation*}
where $c_1$ and $c_2$ are the constants guaranteed by Lemma~\ref{c1c2lemma}.
\end{lemma}

A constructed equilibrium candidate is feasible if it meets the following conditions:
\begin{enumerate}
    \item For $\texttt{type}=$ I.A.ii, there exists a value of $\alpha_{j^{[2]}}\in(0,1)$ such that:
        \begin{enumerate}
            \item The values of $\alpha_i,\beta_i$ for $i\in I_5$ given by \eqref{alphabetaI5} lie in the interval $(0,1)$.
            \item The values of $c_1,c_2$ given by Table~\ref{c1c2table} meet the conditions of Lemma~\ref{c1c2feasibilitylemma}.
            \item The constant sum condition \eqref{sum_alpha_beta} is met.
        \end{enumerate}
    \item For $\texttt{type}=$ I.A.iii, there exists a value of $\alpha_{j^{[8]}}\in (0,1)$ such that the conditions 1.(a),(b),(c) above are met.
    \item For $\text{type}=$ I.B.i, there exists a value of $\beta_{j^{[6]}}\in (0,1)$ such that the conditions 1.(a),(b),(c) above are met.
    \item For $\texttt{type}=$ I.B.ii or $\texttt{type}=$ I.B.iii, the values of $\alpha$ and $\beta$ given by \eqref{IBiialphabeta} and \eqref{IBiiialphabeta} (respectively) lie in the interval $(0,1)$ and condition 1.(b) above is met.
    \item For $\texttt{type}=$ I.A.i, the conditions 1.(a) and 1.(b) above are met.
\end{enumerate}

\begin{algorithm} Equilibrium Computation\\
    \begin{algorithmic}
        \State {\bf Input:} $k_a,k_d,U_a^u,U_a^c,U_d^u,U_d^c$
        \For{$r=0:\min\{m-k_a,m-k_d\}$}
            \For{$s=0:\min\{k_a,m-k_d-r\}$}
                \For{$t=0:\min\{k_a-s,k_d\}$}
                    \For{$\texttt{type}\in\{$I.A.i,I.A.ii,I.A.iii,I.B.i,I.B.ii,I.B.iii$\}$}
                        \If{$\Gamma(r,s,t,\texttt{type})$ is feasible}
                            \State {\bf Output} $(\alpha^*,\beta^*)$ and $(v_a^*,v_d^*)$
                        \EndIf
                    \EndFor
                \EndFor
            \EndFor
        \EndFor
    \State {\bf Else:} No Type I equilibrium exists.
    \end{algorithmic}
    \label{equilibriumalgorithm}
\end{algorithm}

For a given $r,s,t$ and $\texttt{type}$, let $\Gamma(r,s,t,\texttt{type})$ denote the data of the candidate equilibrium corresponding to the given parameters constructed according to the process just outlined. The Type I equilibrium computation procedure is given by Algorithm~\ref{equilibriumalgorithm}. 

\begin{lemma}
    Under Assumption~\ref{distinctassumption}, if $k_d\leq k_a$ no Type II equilibrium exists.    
\end{lemma}
\begin{proof}
    First note that equation \ref{sum_alpha_beta} must hold, and therefore $I_8=\emptyset$. Then, note that in such a type II equilibrium, we have $\sum_{i=1}^m\alpha_i=\sum_{i\in I_9}1=k_a$. Thus, if $k_d\leq k_a$ either $I_4\cup I_7\neq \emptyset$ or $\sum_{i=1}^m\beta_i>k_d$ and no type II equilibrium exists.
\end{proof}

If $k_d>k_a$, a Type II equilibrium may exist.
By Observation~\ref{theobservation}, we have $U_a^c(i)<U_a^c(j)$ for $i\in I_1\cup I_4\cup I_7$ and $j\in I_9$. We order $T$ so that $U_a^c(i)<U_a^c(j)$ for $i<j$ and set $I_9=\{m-k_a+1,\dots,m\}$. Let $I_7$ contain any $k_d-k_a$ targets, and let $I_1$ contain any remaining targets. Note that $c_2=0$ and $c_1=U_d^c(m-k_a+1)$ testifies that this constructed $(\alpha,\beta)$ is a Nash equilibrium.

\begin{example}
    It is straightforward to construct a game exhibiting any given type of equilibrium and feasible parameters $r,s,t$. For example, fix some $k_a$ and $k_d$ and constants $c_1,c_2$. Let $n=\max\{k_a,k_d\}+1$. Construct a game on  $T=\{1,\dots,n\}$ by taking $x_1,\dots,x_n\in (0,1)$ to be distinct reals such that $\sum_{i=1}^nx_i=k_a$ and then picking any $U_d^u(i),U_d^c(i)$ such that $\Delta_d(i)=x_i$ for each $i$. Then, take $y_1,\dots,y_n\in (0,1)$ such that $\sum_{i=1}^ny_i=k_d$ and select $U_a^u(i),U_a^c(i)$ distinct and such that $U_a^u(i)-y_i\Delta_a(i)=c_1$ for each $i$. The resulting game has a type I.A.i with $r=s=t=0$ and constants $c_1$ and $c_2$. 
    
    As an explicit example, take $k_a=3$, $k_d=2$ and $c_1=1$, $c_2=2$. Then $n=4$ and we take $x_1=0.6,x_2=0.7,x_3=0.9,x_4=0.8$ and $y_1=0.3,y_2=0.5,y_3=0.4,y_4=0.8$. We choose
    \begin{center}
        \begin{tabular}{|c|c|c|c|c|}
            \hline
            Target:&1&2&3&4\\
            \hline
            $U_a^u$&8/7&6/5&4/3&2\\
            \hline
            $U_a^c$&2/3&4/5&1/2&3/4\\
            \hline
            $U_d^u$&-8/5&-27/10&-39/10&-24/5\\
            \hline
            $U_d^c$&-1&-2&-3&-4\\
            \hline
        \end{tabular}
    \end{center}
    
    Clearly, one can introduce new targets meeting the criteria in Lemmas~\ref{attackerstructurallemma} and~\ref{defenderstructurallemma} to construct an equilibrium of type I.A.i with other values of $r,s,t$. Furthermore, from a type I.A.i equilibrium with $r=s=t=0$, we can build equilibria of types I.A.ii,I.A.iii,I.B.i,I.B.ii and I.B.iii by introducing targets which belong to $I_2,I_8$ or $I_6$ by enforcing that $U_a^u(i),U_a^c(i)$ or $\Delta_d(i)$ coincide with the fixed value of $c_1$ or $c_2$ respectively.
\end{example}

\begin{remark}
    For the sake of clarity of presentation, the analysis in this section has been under Assumption~\ref{distinctassumption}. However, it is straightforward to extend our analysis and modify Algorithm~\ref{equilibriumalgorithm} to cases in which Assumption~\ref{distinctassumption} does not hold. When Assumption~\ref{distinctassumption} is lifted, we may have $|I_2|,|I_6|,|I_8|>1$ and it may be the case that $I_2$ and $I_8$ are both nonempty. Thus, two additional types of equilibrium, I.A.iv and I.B.iv are introduced to account for the cases in which $I_2\neq\emptyset$ and $I_8\neq \emptyset$. By Lemma~\ref{c1c2lemma}: If $I_2\neq \emptyset$, we must have $U_a^u(i)=U_a^u(j)$ for all $i,j\in I_2$; if $I_8\neq \emptyset$, $U_a^c(i)=U_a^c(j)$ we must have for all $i,j\in I_8$; and if $I_6\neq\emptyset$ we must have $\Delta_d(i)=\Delta_d(j)$ for all $i,j\in I_6$. Furthermore, for any Type I.A.iv or I.B.iv, we must have $U_a^u(i)=U_a^c(j)$ for $i\in I_2$ and $j\in I_8$. When payoffs are possibly indistinct, we modify Algorithm~\ref{equilibriumalgorithm} to include these cases. Note that when checking for feasibility of any Type I.A.ii,I.A.iii,I.A.iv or I.B.i we must verify that there exists a collection of $\alpha_i$ or $\beta_i$ (respectively) in $(0,1)$ meeting the condition \eqref{sum_alpha_beta} and such that $\alpha_j,\beta_J$ for $j\in I_5$ belong to $(0,1)$. Additionally, when payoffs are possibly indistinct, a Type II equilibrium may have $I_8\neq \emptyset$. We check for feasible Type II equilibria by introducing a parameter $\ell=|I_9|$ and considering, for each $\ell$, all candidate equilibria in which $U_a^c(i)=U_a^c(j)$ for all $i,j\in I_8$.
\end{remark}

As the structural analysis of section 2.1 gives closed-form expressions for $\alpha$ and $\beta$, we obtain as a secondary product a set of closed-form expressions for the player expected outcomes at equilibrium for each type. These expressions are presented in the appendix.

\subsection{Uniqueness and Multiplicity of Nash Equilibria}

In this section, we present results regarding the uniqueness and multiplicity of equilibria of the various types defined in Section 2.1. The following result is due to \cite{korzhyk2011security}
\begin{lemma}\label{korzhyk_lemma}
    Suppose $(\alpha^*,\beta^*)$ and $(\overline{\alpha}^*,\overline{\beta}^*)$ are two Nash equilibria of a security game.
    \begin{enumerate}[(i)]
        \item For all $t\in T$, $\alpha^*_t=\overline{\alpha}^*_t$ or $\beta^*_t=\overline{\beta}^*_t$
        \item If $t_1,t_2\in T$ are such that $\beta^*_{t_1}\neq \overline{\beta}^*_{t_1}$ and $\beta^*_{t_2}\neq \overline{\beta}^*_{t_2}$ then \[\alpha^*_{t_1}\Delta_d(t_1)=\alpha^*_{t_2}\Delta_d(t_2)=\overline{\alpha}^*_{t_1}\Delta_d(t_1)=\overline{\alpha}^*_{t_2}\Delta_d(t_2)\]
    \end{enumerate}
\end{lemma}
From this result, we derive the following corollary:
\begin{corollary}\label{c1c2corollary}
    If $(\alpha^*,\beta^*)$ and $(\overline{\alpha}^*,\overline{\beta}^*)$ are two Nash equilibria of a security game with constants $c_1,c_2$ and $\overline{c_1},\overline{c_2}$ respectively then one of the following holds:
    \begin{enumerate}[(i)]
        \item $c_1=\overline{c_1}$ and $c_2=\overline{c_2}$
        \item Precisely one of $c_1=\overline{c_1}$ or $c_2=\overline{c_2}$ holds.
    \end{enumerate}
\end{corollary}
\begin{proof}
    Note that if $\alpha^*=\overline{\alpha}^*$ and $\beta^*=\overline{\beta}^*$, the two equilibria are of precisely the same type with the same parameters and thus $c_1=\overline{c_1}$, $c_2=\overline{c_2}$. Suppose that there exists $t$ such that $\alpha_t>\overline{\alpha}_t$. By Lemma~\ref{korzhyk_lemma} (i), we have $\beta_i=\overline{\beta}_i$. Therefore \[c_1\leq U_a^c(t)\beta_i+U_a^u(t)(1-\beta_t)=U_a^c(t)\overline{\beta}_i+U_a^u(t)(1-\overline{\beta}_t)\leq \overline{c_1}.\] Now, as $\sum_{t=1}^m\alpha_i=k_a$, there exists $t^\prime$ so that $\alpha_{t^\prime}<\overline{\alpha}_{t^\prime}$. By the same argument as that for $t$, we have $\overline{c_1}\leq c_1$ and thus $c_1=\overline{c_1}$. A similar argument shows $c_2=\overline{c_2}$ if there exists $t$ such that $\beta_t=\overline{\beta}_t$.
\end{proof}

\begin{lemma}\label{multiplicitylemma}
    \begin{enumerate}[(i)]
        \item If multiple Type I equilibria exist, they are all of the same subtype
        \item If an equilibrium of type I.A.i, I.B.ii or I.B.iii exists, it is unique
        \item If an equilibrium of type I.A.ii,I.A.iii or I.B.i exists, a continuum of such equilibria exist
    \end{enumerate}
\end{lemma}
\begin{proof}
    Suppose $(\alpha^*,\beta^*),(\overline{\alpha}^*,\overline{\beta}^*)$ are two type I equilibria. By Corollary~\ref{c1c2corollary}, $\alpha^*\neq \overline{\alpha}^*$ implies $c_1=\overline{c_1}$ and $\beta^*\neq\overline{\beta}^*$ implies $c_2=\overline{c_2}$. Consider the cases:
    \begin{itemize}
        \item $c_1=\overline{c_1}$ and $c_2=\overline{c_2}$; We have $(\alpha^*,\beta^*)=(\overline{\alpha}^*,\overline{\beta}^*)$.
        \item $c_1=\overline{c_1}$ and $c_2\neq\overline{c_2}$; There must exist $t$ so that $c_1=\overline{c_1}=U_a^u(t)-\beta_i^*\Delta_a(t)$ and thus there exists $t\in I_2\cup I_8$. 
        \item $c_1\neq\overline{c_1}$ and $c_2=\overline{c_2}$; There must exist $t$ so that $c_2=\overline{c_2}=\alpha_t\Delta_d(t)$ and thus there must exist $t\in I_6$.
    \end{itemize}
    That is, if two distinct equilibria of Type I exist they are of type $I.A.ii,I.A.iii$ or $I.B.i$. Thus, if a Type I.A.i, I.B.ii or I.B.iii equilibrium exists, it is unique. Furthermore, by Assumption~\ref{distinctassumption} if two distinct equilibria exist of types I.A.ii or I.A.iii, they must both be either of type I.A.ii or I.A.iii.
\end{proof}
\begin{lemma}
    If an equilibrium of Type II exists, no equilibrium of Type I exists.
\end{lemma}
\begin{proof}
    Suppose $(\alpha^*,\beta^*)$ is of Type II and $(\overline{\alpha}^*,\overline{\beta}^*)$ is of Type I. As there is a feasible Type II equilibrium, we have $k_d>k_a$. By the interchangeability of Nash equilibria in additive security games \cite{korzhyk2011stackelberg}, we have that $(\overline{\alpha}^*,\beta^*)$ must be a Nash equilibrium. If $(\overline{\alpha}^*,\beta^*)$ is of Type I, it must be of Type I.A.i with $|I_9|=k_d$ as the defender is playing a pure strategy. However, this implies $\sum_{i=1}^m\alpha_i\geq \sum_{i\in I_9}\alpha_i=k_d>k_a$ in violation of $\eqref{sum_alpha_beta}$. Thus $(\overline{\alpha}^*,\beta)$ is of type II, so $\overline{\alpha}^*$ is a pure strategy. On the other hand, if $(\alpha,\overline{\beta}^*)$ is of Type I, it must be of Type $I.A.i$ with $I_5=\emptyset$ and $|I_9|=k_d$. However this implies $k_a=\sum_{i=1}^m\alpha_i\geq \sum_{i\in I_3\cup I_9}1\geq k_d$ contradicting $k_d>k_a$. So $(\alpha,\overline{\beta^*})$ is of Type II and $\overline{\beta}^*$ is a pure strategy. Therefore $(\overline{\alpha}^*,\overline{\beta}^*)$ is a Type I.A.i with $I_5=\emptyset$ (as both players are playing pure strategies), but this implies $k_d\leq k_a$ - a contradiction.
\end{proof}

\subsection{Special Case: Fully Protective Resources}

In this section we consider the special case of {\it fully protective resources}. That is, $U_a^c(i)=U_d^c(i)=0$ for all $i$. A security game model in which $U_a^c(i)=0$ for all $i$ and the attacker has a single resource is studied in \cite{lerma2011linear} in which a linear time algorithm algorithm was proposed to calculate a Stackelberg equilibrium in such a game.

As in the previous section, we shall make Assumption~\ref{distinctassumption} in the following analysis, and we present closed-form expressions for the player expected payoffs at equilibrium in the appendix. We now present the analogues of Lemmas~\ref{I4I7lemma}-\ref{attackerstructurallemma} and Theorem~\ref{NEtheorem} for the case of fully protective resources. Note that, for any pair of marginal attack and defense probability vectors $(\alpha,\beta)$, we have the following, even more compact forms of the defender expected payoffs \eqref{vavd}:
\[v_a^*=\sum_{t=1}^m\alpha_tU_a^u(t)(1-\beta_t)\hspace{1em}\text{and}\hspace{1em}v_d^*=\sum_{t=1}^m\alpha_iU_d^u(t)(1-\beta_t).\]

Intuitively, the expected payoffs depend only upon the marginal probabilities with which the targets are attacked and exposed. The structural characterization of equilibria in the case of fully protective resources will, as in the previous section, from repeated application of the following modified version of Observation~\ref{theobservation}
\begin{observation}\label{fullyprotectiveobservation}
    In a game with fully protective resources, a player can unilaterally deviate from $(\alpha,\beta)$ if and only if they are able to shift marginal probability from one target to another. By definition, therefore, if $(\alpha^*,\beta^*)$ is a Nash equilibrium:
    \begin{equation*}
    \begin{aligned}
        \alpha_i^*>0,\alpha_j^*<1&\implies U_a^u(j)(1-\beta_j)\leq U_a^u(i)(1-\beta_i)\\
        \beta_i^*>0,\beta_j^*<1&\implies \alpha_iU_d^u(i)\leq\alpha_jU_d^u(j)
    \end{aligned}
    \end{equation*}
    That is, no player can unilaterally deviate and increase their expected payoff.
\end{observation}
\begin{lemma}\label{exclusionlemma}
    If $(\alpha^*,\beta^*)$ is a Nash equilibrium of a security game with fully protective resources:
    \begin{enumerate}[(i)]
        \item $I_4\cup I_7=\emptyset$
        \item At least one of $I_8\cup I_9$ and $I_1\cup I_2\cup I_5$ is empty. 
        \item At least one of $I_8\cup I_9$ and $I_6$ is empty.
    \end{enumerate}
\end{lemma}

\begin{theorem}\label{NEtheoremfullyprotectiveresources}
    A Nash equilibrium $(\alpha^*,\beta^*)$ in a security game with fully protective resources is of one of the following types:
    \begin{enumerate}[(i)]
        \item I.A.i, I.A.ii,I.B.i or I.B.ii with $I_9=\emptyset$
        \item I.A.iii with $I_1\cup I_5=\emptyset$
    \end{enumerate}
\end{theorem}

\begin{lemma}
    If $(\alpha^*,\beta^*)$ is a Nash equilibrium in a security game with fully protective resources:
    \begin{enumerate}[(i)]
        \item There exists a constant $c_1$ such that $c_1=U_a^u(i)(1-\beta_i^*)$ for all $i\in I_2\cup I_5$. Furthermore, for $i\in I_2$ we have $c_1=U_a^u(i)$.
        \item There exists a constant $c_2$ such that $c_2=\alpha_i^*U_d^u(i)$ for all $i\in I_5\cup I_6$. Furthermore, for $i\in I_6$ we have $c_2=U_d^u(i)$.
    \end{enumerate}
\end{lemma}

\begin{lemma}
    Suppose $(\alpha^*,\beta^*)$ is a Nash equilibrium in a security game with fully protective resources. Then
    \begin{enumerate}[(i)]
        \item For $i\in I_2\cup I_3\cup I_5\cup I_6$ and $j\in I_1$, $U_a^u(j)<U_a^u(i)$.
        \item For $i\in I_3\cup I_6$, $j\in I_2$, $U_a^u(j)<U_a^u(i)$.
        \item For $i\in I_3\cup I_6$, $j\in I_5$, $U_a^u(1-\beta_j)\leq U_a^u(i)$.
        \item For $i\in I_5\cup I_6$, $j\in I_2$, $U_d^u(i)\leq \alpha U_d^u(j)$.
        \item For $i\in I_5\cup I_6$, $j\in I_3$, $U_d^u(i)<U_d^u(j)$.
    \end{enumerate}
    \label{fullyprotectivestructurallemma}
\end{lemma}

\begin{table}
\caption{Expressions for $c_1$ and $c_2$ for candidate equilibria in the case of fully protective resources.}
\begin{center}
    \setlength{\tabcolsep}{0.5em} 
    \renewcommand{\arraystretch}{2.7}
    \begin{tabular}{|c|c|c|}
        \cline{2-3}
        \multicolumn{1}{c|}{} & $c_1$ & $c_2$\\ \hline
        I.A.i with $I_9=\emptyset$&${k_d-m+r+s\over\sum_{j\in I_5}{1\over U_a^u(j)}}$&${k_a-s\over \sum_{j\in I_5}{1\over U_d^u(i)}}$\\ \hline
        I.A.ii with $I_9=\emptyset$&$U_a^u(j^{[2]})$&${k_a-s-\alpha_{j^{[2]}}\over \sum_{j\in I_5}{1\over U_d^u(j)}}$\\ \hline
        I.B.i with $I_9=\emptyset$&${k_d-m+r+s+1-\beta_{j^{[6]}}\over \sum_{j\in I_5}{1\over U_a^u(j)}}$&$U_d^u(j^{[6]})$\\ \hline
        I.B.ii with $I_9=\emptyset$&$U_a^u(j^{[2]})$&$U_d^u(j^{[6]})$\\ \hline
    \end{tabular}
\label{c1c2tablefullyprotectiveresources}
\end{center}
\end{table}

\begin{corollary}
     In the case of fully protective resources, Algorithm~\ref{equilibriumalgorithm} can be modified to be $O(m^2)$.
\end{corollary}
\begin{proof}
    By the result in Theorem~\ref{NEtheoremfullyprotectiveresources}, it suffices to consider two cases. Suppose $(\alpha^*,\beta^*)$ is of type I.A.i, I.A.ii,I.B.i or I.B.ii with $I_9=\emptyset$. In each of these types we have $I_9=\emptyset$ and therefore it suffices to run Algorithm~\ref{equilibriumalgorithm} without iteration over $t$. The values of $c_1$ and $c_2$ for a given $r,s$ or $s,t$ and type in the case of fully protective resources are given in Table~\ref{c1c2tablefullyprotectiveresources}, and we take these values in the modified version of Algorithm~\ref{equilibriumalgorithm}. To determine the existence of $(\alpha^*,\beta^*)$ of type I.A.iii with $I_1\cup I_5=\emptyset$, first note that we have $r=0$. By Observation~\ref{fullyprotectiveobservation}, $i\in I_8\cup I_9$ and $j\in I_3$ implies $U_d^u(i)<U_d^u(j)$, so we order the targets such that $U_d^u(i)<U_d^u(j)$ for $i<j$. By \eqref{sum_alpha_beta}, then, we must have $I_3=\{k_d+1,\dots,m\}$ and so $\alpha_i=1$, $\beta_i=0$ for $i=k_d+1,\dots,m$. For all remaining $i=1,\dots,k_d$ we have $\beta_i=1$ and $\alpha_1,\dots,\alpha_{k_d}$ must satisfy \[\sum_{i\in T\setminus I_3}\alpha_i=k_a-m+k_d+2\hspace{1em}\text{and}\hspace{1em}{U_d^u(k_d+1)\over U_d^u(i)}\leq \alpha_i\leq 1\] by \eqref{sum_alpha_beta} and Lemma~\ref{c1c2lemma}. There is a solution if and only if $\sum_{i\in T\setminus I_3}{U_d^u(k_d+1)\over U_d^u(i)}\leq k_a-m+k_d+2$. 
    
\end{proof}

\subsection{Zero Sum Games with Fully Protective Resources}

In this section, we consider the case of zero-sum security games with fully protective resources. That is, $U_a^c(i)=U_d^c(i)=0$ and $U_a^u(i)=-U_d^u(i)$ for all $i$. By von Neumann's theorem \cite{neumann1928theorie}, we have
\[v^*=\max_p\min_{1\leq i\leq{m\choose k_d}}(p^TA)_i=\min_q\max_{1\leq j\leq {m\choose k_a}}(Aq)_j\] where $A=-B$ is the attacker payoff matrix and $p,q$ represent player mixed strategies
Assume without loss of generality that $U_a^u(i)\leq U_a^u(j)$ for $i<j$. Note that \[(p^TA)_i=\sum_{j=1}^{m\choose k_a}p_ja_{ji}=\sum_{t\notin s_i^d}\alpha_tU_a^u(t)\] where $\alpha_t$ is the marginal probability with which $t$ is attacked (as in previous sections). The above expression implies that $v^*$ is the sum of the least $m-k_d$ terms $\alpha_tU_a^u(t)$ at equilibrium. For an equilibrium of type $I.A.i$ with $I_9=\emptyset$, we can introduce a parameterization such that $I_1=\{1,\dots,r\}$, $I_3=\{r+1,\dots,m-s\}$ and $I_5=\{m-s+1,\dots,m\}$. For an equilibrium in which $I_2$ or $I_6$ is nonempty, we take $j^{[2]}=r+1$ or $j^{[6]}=m-s+1$. By Lemma~\ref{fullyprotectivestructurallemma} (and $U_a^u(i)=-U_d^u(i)$) we obtain:
\begin{enumerate}
    \item $\alpha^*U_a^u(m)=\dots=\alpha^*_{m-s+1}U_a^u(m-s+1)=c_2>\alpha^*_{m-s}U_a^u(m-s)\geq\dots\geq \alpha^*_{r+1}U_a^u(r+1),\alpha^*(r)=\dots=\alpha^*(1)=0$
    \item $U_a^u(r+1)\geq(1-\beta^*)U_a^u(m)=\dots=(1-\beta^*_{m-s+1})U_a^u(m-s+1)=c_1>U_a^u(r),\beta^*_{m-s-1}=\dots=\beta^*_1=0$
\end{enumerate}
This structural property coincides with the result obtained in \cite{hamidgamesec20}.
Now, define a function
\[\sigma_\alpha(r,s,\alpha_{r+1},c_2)=m-(s+r+1)+\alpha_{r+1}-\sum_{k=m-s+1}^m{c_2\over U_a^u(k)}.\] By construction, we have $\sigma_\alpha(r,s,\alpha_{r+1},c_2)=\sum_{i=1}^m\alpha_i$. Now, for each $r$ and $s$ note that $c_2\in[-U_a^u(m-s+1),-U_a^u(m-s))$, $\alpha_{r+1}\in[0,1]$ and that $\sigma_\alpha(r,s,\alpha_{r+1},c_2)$ is a strictly non-increasing function of $r$. Thus, \[\sigma_\alpha^{\min}(r,s)\equiv\min_{c_2,\alpha_{r+1}}\sigma_\alpha(r,s,\alpha_{r+1},c_2)=\sup_{c_2,\alpha_{(r+1)+1}}\sigma_\alpha(r+1,s+1,\alpha_{(r+1)+1},c_2)\equiv\sigma_\alpha^{\sup}(r+1,s+1).\] This implies that when $r^\prime>r,s^\prime>s$, if $k_a>\sigma_\alpha^{\min}(r,s)$, $k_a>\sigma_\alpha^{\sup}(r^\prime,s^\prime)$. Thus we can search for feasible equilibria by considering $s\in\{1,\dots,m\}$ and only values of $r$ for which $\sum_{i=1}^m\alpha_i\leq k_a$, yielding a linear time procedure to compute an equilibrium. 

When resources are not fully protective, analogues of properties 1 and 2 above do not hold, as such relationships are dependent upon the ordering of $\{U_a^c(1),\dots,U_a^c(m)\}$ with respect to $\{U_a^u(1),\dots,U_a^u(m)\}$ and thus the assumption of a zero-sum game does not immediately yield a reduction in the complexity of equilibrium computation.

\section{Optimization of Defender Expected Utility}
In this section, we address the problem: Given that a defender can perturb the payoffs of a security game prior to play, how should they do so to optimize their expected outcome at equilibrium? That is, we wish to address the following optimization problem:
\begin{problem}\label{optimizationproblem}
\begin{equation}
\begin{aligned}
    \mathbf{maximize}&\hspace{1em}v_d^*\\
    \mathbf{s.t.}&\hspace{1em}\forall t\in T\hspace{1em}U_a^c(i)\in[lb_a^c(i),ub_a^c(i)], U_a^u(i)\in[lb_a^u(i),ub_a^u(i)]
\end{aligned}
\end{equation}
\end{problem}

\begin{definition}
    In a leader-follower security game in which the Defender is the leader, we say a pair $(\alpha,\beta)$ is a strong Stackelberg equilibrium if and only if both $\alpha$ and $\beta$ are best responses and $\alpha$ yields maximal defender expected outcome among all attacker best responses.
\end{definition}
In the case of Stackelberg equilibria, a single attacker resource, and the restriction that the defender payoffs remain fixed, a polynomial time algorithm to solve Problem~\ref{optimizationproblem} is given in \cite{shi2018designing} (Theorem 5). By \cite{korzhyk2011stackelberg}, in the case of a single attacker resource, such a Stackelberg equilibrium is also a Nash equilibrium and thus a solution to Problem~\ref{optimizationproblem} for Nash equilibria can also be found in polynomial time. However, in the case of multiple attacker resources, a Stackelberg equilibrium need not be a Nash equilibrium (as shown in \cite{korzhyk2011stackelberg}).
Consider the following restricted version of Problem~\ref{optimizationproblem}:
\begin{problem}\label{restrictedoptimizationproblem}
    \begin{equation}
\begin{aligned}
    \mathbf{maximize}&\hspace{1em}v_d^*\\
    \mathbf{s.t.}&\hspace{1em}\forall t\in T\hspace{1em}U_a^c(i)\in\{lb_a^c(i),ub_a^c(i)\}, U_a^u(i)\in\{lb_a^u(i),ub_a^u(i)\}
\end{aligned}
\end{equation}
\end{problem}

That is, Problem~\ref{restrictedoptimizationproblem} restricts Problem~\ref{optimizationproblem} by allowing each payoff to take only one of two values, rather than all values in some closed interval.

\begin{theorem}\label{stackelberghardness}
    In the case of Stackelberg equilibrium and multiple attacker resources, Problem~\ref{restrictedoptimizationproblem} is weakly NP-hard.
\end{theorem}
\begin{proof}
    Define a problem {\bf TwoWeightKnapsack} as follows: given a collection of $k$ objects each with a value $v_i$ and a set of two possible weights $\{w_i^1,w_i^2\}$ determine if there exists a subset of the objects $S$ and a choice of $w_i=w_i^1$ or $w_i=w_i^2$ for each $i$ such that \[\sum_{t\in S}w_t\leq 1\text{   and   }\sum_{t\in S}v_t\geq V.\] Clearly, a standard instance of {\bf Knapsack} can be reduced to an instance of {\bf TwoWeightKnapsack} by taking $w_i^1=w_i^2$. We adapt a construction presented in \cite{korzhyk2011security}. Suppose $v_i,\{w_i^1,w_i^2\},V$,$i=1,\dots,k_a$ is an instance of {\bf TwoWeightKnapsack}. Define a security game on $2k_a$ targets in which $k_d=1$. Take targets $1,\dots,k$ to be the objects in the instance of {\bf TwoWeightKnapsack}. For a given choice of $w_t=w_t^1$ or $w_t=w_t^2$ for each $t$, $1\leq i\leq k$ take \[U_a^u(i)=w_i,\hspace{1em} U_a^c(i)=w_i-1,\text{  and  } U_d^c(i)=U_d^c=-v_i.\] Let targets $k+1,\dots,2k$ be such that \[U_a^u(i)=U_a^c(i)=U_d^c(i)=U_d^u(i)=0.\] Note that the defender can obtain an expected utility of at least $V-\sum_{i=1}^kv_i$ if and only if there is a solution to the instance of {\bf TwoWeightKnapsack}.
\end{proof}
\begin{corollary}
    In the case of Stackelberg equilibria and multiple attacker resources, Problem~\ref{optimizationproblem} is weakly NP-hard.
\end{corollary}
The hardness of Problem~\ref{restrictedoptimizationproblem} in the case of Stackelberg equilibria is fundamentally due to the weak NP-hardness of computing Stackelberg equilibria in security games with multiple attacker resources. The problem of computing Nash equilibria in such games, however, can be solved in polynomial time (as demonstrated by \cite{korzhyk2011security} or the algorithm of Section 2.1). The nontrivial nature of Problem~\ref{restrictedoptimizationproblem} in the case of Nash equilibria is illustrated by the following example:

\begin{example}
    Consider the instance of Problem~\ref{restrictedoptimizationproblem} defined by the following zero-sum game with fully protective resources, $k_a=2$ and $k_d=3$:
    \begin{center}
        \begin{tabular}{|c|c|c|c|c|c|c|}
            \hline
            $U_d^u(i)$&-5&-10&-7&-8&-4&-1\\ \hline
            $lb_a^u(i)$&1&2&9&4&6&10\\ \hline
            $ub_a^u(i)$&7&3&13&5&8&11\\ \hline
        \end{tabular}
    \end{center}
    If $U_a^u(i)=lb_a^u(i)$ for all $i$ or $U_a^u(i)=ub_a^u(i)$ for all $i$, we have $v_d^*=-789/229\approx-3.45$ with $\alpha^*=[56/229,28/229,40/229,35/229,70/229,1]$ in both cases and $\beta_{lb}^*=[1/73,37/73,65/73,$ $55/73,61/73,0]$, $\beta_{ub}^*=[6469/9589,2309/9589,7909/9589,5221/9589,6859/9589,0]$ respectively. However, one can verify by exhaustive search that the maximum value of $v_d^*$, $v_d^*=-453/173\approx-2.62$, is attained when (among other choices) $U_a^u(1)=lb_a^u(1)$ and $U_a^u(i)=ub_a^u(i)$ for all other $i$ with equilibrium strategies: 
    \begin{center}
        \begin{tabular}{|c|c|c|c|c|c|c|}
            \hline
            $\alpha^*$&0&28/173&40/173&35/173&70/173&1\\ \hline
            $\beta^*$&0&627/1147&1027/1147&835/1147&952/1147&0\\ \hline
        \end{tabular}
    \end{center}
\end{example}

Given the nontrivial nature of the problem, we introduce the following assumption enforcing that all parameters of distinct targets belong to distinct intervals, which will allow for solution of problem~\ref{restrictedoptimizationproblem}.

\begin{assumption}\label{disjointintervalsassumption}
 \[lb_a^c(1)\leq ub_a^c(1)<lb_a^c(2)\leq ub_a^c(2)<\dots<lb_a^c(m)\leq ub_a^c(m)\]
    and that there exists a permutation $i_1,\dots,i_m$ of $[m]$ such that
    \[lb_a^u(i_1)\leq ub_a^u(i_1)<lb_a^u(i_2)\leq ub_a^u(i_2)<\dots<lb_a^u(i_m)\leq ub_a^u(i_m).\]
\end{assumption}

 We now show that there is a pseudopolynomial time algorithm to solve Problem~\ref{restrictedoptimizationproblem} for the case of Nash equilibria under Assumptions~\ref{distinctassumption} and~\ref{disjointintervalsassumption}.
 
\begin{theorem}\label{pseudopolynomialtheorem}
    Under Assumptions~\ref{distinctassumption} and~\ref{disjointintervalsassumption}, in the case of multiple attacker resources and Nash equilibria, Problem~\ref{restrictedoptimizationproblem} can be solved in pseudopolynomial time.
\end{theorem}
\begin{proof}
    Note that by our structural characterization of Nash equilibria given in Section 2.1, Nash equilibria are of precisely three classes: equilibria in which $c_1$ and $c_2$ both coincide with a parameter of the game (I.B.ii,I.B.ii), those in which precisely one of $c_1$ or $c_2$ coincides with a game parameter (I.A.ii,I.A.iii,I.B.i), and those in which neither $c_1$ nor $c_2$ coincide with a game parameter (I.A.i). We now present a procedure to identify the equilibrium of largest $v_d^*$ within each class.
    Consider the cases:
    \begin{enumerate}[(i)]
        \item $\exists i\exists j$ so $c_1=U_a^x(i)$ and $c_2=\Delta_d(j)$ for some $x\in\{u,c\}$. First consider the problem of computing an equilibrium of type I.B.ii of maximal $v_d^*$. We have that $c_1$ and $c_2$ correspond to some parameters associated with some pair of targets. There are $2|T|$ choices for $c_1$ and $|T|$ choices for $c_2$. Fix some choice of $c_1$ and $c_2$. For any type I.B.ii equilibrium, there exists some $\ell\in\{-\infty, U_a^u(1),\dots, U-a^u(m)\}$ such that $t\in I_1$ if and only if $U_a^u(t)\leq \ell$. Note there are $2|T|+1$ choices of $\ell$. Fix some choice of $\ell\in\{-\infty,lb_a^u(i_1),\dots,ub_a^u(i_m)\}$. If $\ell=lb_a^u(i)$ or $\ell=ub_a^u(i)$ we take $U_a^u(i)=\ell$. Note that for any $t\in I_1$, $v_d^*$ does not depend upon the choice of $U_a^u(t)$ or $U_a^c(t)$ (as $\alpha_t=0$ for $t\in I_1$), so we pick $U_a^u(t),U_a^c(t)$ for $t\in I_1$ arbitrarily. Assume that $\ell$ selects $r$ targets to belong to $I_1$. We now consider the possible distributions of the remaining $m-r-2$ targets to the sets $I_3,I_5$ and $I_9$. First note that as $I_2\cup I_6\neq \emptyset$ we must have $I_5\neq \emptyset$. Recall the notation $s=|I_3|$ and $t=|I_9|$ from section 2.2. By Lemma~\ref{defenderstructurallemma}, for each possible choice of $r$ and $s$ we take $I_3$ to consist of the $s$ targets in $T\setminus I_1\cup I_2\cup I_6$ of least $\Delta_a(i)$ and $I_9$ to consist of the $t$ targets in $T\setminus I_1\cup I_2\cup I_3\cup I_6$ of greatest possible choice of $U_a^c(i)$. Finally, we take $I_5$ to consist of all remaining targets. For $t\in I_5$, we must have $\alpha_t=c_2/\Delta_d(t)$ and $\beta_t=(U_a^u(t))/\Delta_d(t)$. If any target in $I_5$ has $\Delta_d(t)$ such that $\alpha_t\notin (0,1)$, there is no feasible I.B.ii equilibrium with the current $c_1,c_2,r,s,t$. For each target $t\in I_5$ we exclude any choice of $lb_a^u(t)$ or $ub_a^u(t)$ which violates $\beta_t\in (0,1)$. If both such choices violate $\beta_t\in (0,1)$ for some $t\in I_5$ there is no feasible I.B.ii equilibrium with the current $c_1,c_2,r,s,t$. Additionally, we must have $\alpha_{j^{[2]}}=k_a-s-t-1-\sum_{i\in I_5}c_2/\Delta_d(i)\in (0,1)$ and if this condition is violated there is no feasible I.B.ii equilibrium with the current parameters. To check for feasible candidate solutions among the remaining choices, we compute all possible choices of $lb_a^c(i),ub_a^c(i),lb_a^u(i),ub_a^u(i)$ for which $\beta_{j^{[6]}}=k_a-t+\sum_{i\in I_5}U_a^u(j^{[2]})/\Delta_a(i)-\sum_{i\in I_5}U_a^u(i)/\Delta_a(i)\in (0,1)$. This is an instance of {\bf SubsetSum} which can be solved in pseudopolynomial time using standard dynamic programming techniques \cite{cormen2022introduction}. For each potentially feasible constructed equilibrium, we check the feasibility conditions given in section 2.1. Finally, we select an equilibrium of type I.B.ii of maximal $v_d^*$ or output that no such equilibrium exists. The construction for type I.B.iii equilibria is similar and differs only in that $\{c_1=a_i^c\}=I_8$ for some $i$ and $I_2=\emptyset$.
        \item For all $t\in T$, $U_a^u(t)\neq c_1\neq U_a^u(t)$ and $c_2\neq \Delta_d(i)$. Note that such an equilibrium is of type I.A.i. Either there exists $i\in T$ so that $c_2\in (\Delta_d(i),\Delta_d(i+1)$ or one of $c_2\leq \Delta_d(1)$, $\Delta_d(m)<c_2$ holds. Thus, there are $m+1$ possible choices of interval for $c_2$. Additionally, we have that $c_1$ must lie in an interval strictly between some choice of $U_a^c(i),U_a^u(i),U_a^c(j),U_a^u(j)$ or strictly larger or smaller than all such choices. Clearly there are polynomially many such possible intervals. Fix two intervals $c_1\in(a,b)$ and $c_2\in (c,d)$. As in case $(i)$, we establish the parameters $\ell,s$ and $t$, construct every candidate $I_1,I_3,I_5$ and $I_9$ and eliminate any choices of payoff for $t\in I_5$ that lead to $\alpha_t\notin (0,1)$ or $\beta_t\notin(0,1)$ (or determine the current parameters lead to no feasible equilibrium). We then compute all candidate choices of parameters for targets in $I_5$ such that $c_1=(t-k_d+\sum_{j\in I_5}U_a^u(i)/\Delta_a(j))/\sum_{j\in I_5}1/\Delta_a(j)\in (a,b)$ and $c_2=(k_a-s-t)/\sum_{j\in I_5}1/\Delta_a(j)\in (c,d)$. This amounts to the following modified version of {\bf SubsetSum}: Given a list of numbers $n_1,\dots,n_k$ and an interval $(a,b)$, compute all subsets which have a sum belonging to the interval $(a,b)$. Assuming that our game parameters are represented as floating point numbers, this problem can be solved in pseudopolynomial time by solving an equivalent problem in which all parameters are integers and the target interval is $(a^\prime,b^\prime)$ by way of $|\{i\in (a^\prime,b^\prime): i\in \mathbb{Z}\}|$ instances of {\bf SubsetSum}. Finally, we check each candidate solution for feasibility using the conditions presented in section 2.1 and select an equilibrium of type I.A.i of maximal $v_d^*$ or output that no such equilibrium exists.
        \item Precisely one of $c_1$ or $c_2$ coincides with a game parameter. It is clear from the constructions for cases $(i)$ and $(ii)$ that this case is also resolved in pseudopolynomial time.
    \end{enumerate}
\end{proof}

\begin{figure}
\caption{Decision diagram used to improve the efficiency of the algorithm for type I.B.ii equilibria presented in Theorem~\ref{pseudopolynomialtheorem}}
    \begin{tikzpicture}[font=\scriptsize]
    \tikzstyle{level 1}=[level distance=12mm,sibling distance=60mm]
    \tikzstyle{level 2}=[level distance=15mm,sibling distance=35mm]
    \tikzstyle{level 3}=[level distance=15mm,sibling distance=15mm]
    \tikzset{solid node/.style={circle,draw,inner sep=1,fill=black}}
        \node [solid node]{}
            child {node [solid node]{}
                child {node [solid node]{}
                    child {node [solid node, label=below:{$i\in I_1$}]{} edge from parent node[left]{$\Delta_d(i)<c_2$}}
                    child {node [solid node,label=below:{$i\in I_1$}]{} edge from parent node[right]{$\Delta_d(i)>c_2$}}
                edge from parent node[left]{$U_a^c(i)<c_1$}}
            child {node [solid node]{}
                    child {node [solid node,label=below:{$i\in I_3$}]{} edge from parent node[left]{$\Delta_d(i)<c_2$}}
                    child {node [solid node,label=below:{Infeasible}]{} edge from parent node[right,xshift=-5,yshift=8]{$\Delta_d(i)>c_2$}}
                edge from parent node[right]{$U_a^c(i)>c_1$}}
            edge from parent node[left,yshift=2]{$U_a^u(i)<c_1$}}
            child {node  [solid node]{}
                child {node [solid node]{}
                    child {node [solid node,label=below:{$i\in I_3$}]{} edge from parent node[left,xshift=-1,yshift=-5]{$\Delta_d(i)<c_2$}}
                    child {node [solid node,label=below:{$i\in I_5$}]{} edge from parent node[right]{$\Delta_d(i)>c_2$}}
                edge from parent node[left]{$U_a^c(i)<c_1$}}
            child {node [solid node]{}
                    child {node [solid node,label=below:{$i\in I_3$}]{} edge from parent node[left]{$\Delta_d(i)<c_2$}}
                    child {node [solid node,label=below:{$i\in I_9$}]{} edge from parent node[right]{$\Delta_d(i)>c_2$}}
                edge from parent node[right]{$U_a^c(i)>c_1$}}
            edge from parent node[right,yshift=2]{$U_a^u(i)>c_1$}};
            
    \end{tikzpicture}
    \label{decisiondiagram}
\end{figure}
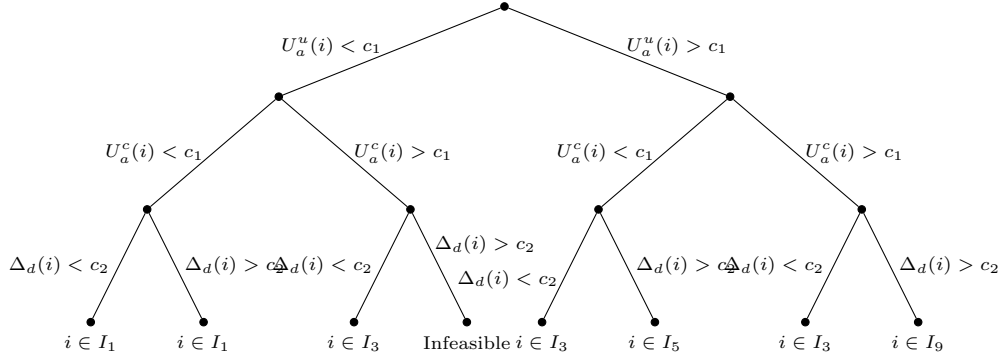

\begin{remark}
    We note that the efficiency of the algorithm presented in the proof of Theorem~\ref{pseudopolynomialtheorem} can be greatly improved by incorporating further paring-down of the search space. For example, the procedure in case (i) for type I.B.ii equilibria may be improved by employing the decision diagram in Figure~\ref{decisiondiagram} to prune the search space.
\end{remark}
\begin{remark}
    It is clear from the proof of Theorem~\ref{pseudopolynomialtheorem} that the pseudopolynomial time algorithm is readily extended to optimization over a domain in which more than two choices are present for each parameter and these choices satisfy an analog of Assumption~\ref{disjointintervalsassumption}.
\end{remark}

\begin{example}
    Consider the following game on five targets with $k_a=3,k_d=2$:
    \begin{center}
        \begin{tabular}{|c|c|c|c|c|c|}
            \hline
            Target:&1&2&3&4&5\\
            \hline
            $lb_a^c(i)$&10&48&5&31&25\\ \hline
            $ub_a^c$(i)&17&49&9&40&29\\ \hline
            $lb_a^u$(i)&20&51&41&63&90\\ \hline
            $ub_a^c$(i)&35&60&42&70&95\\ \hline
            $U_a^c$(i)&-1&-4&-9 &-3&-2\\ \hline
            $U_a^u$(i)&-7&-6&-12&-8&-9 \\ \hline
        \end{tabular}
    \end{center}
    The algorithm described in the proof of Theorem~\ref{pseudopolynomialtheorem} finds
    a feasible choice of parameters yielding the optimal defender expected outcome $v_d=-18$ given by 
    \begin{center}
        \begin{tabular}{|c|c|c|c|c|c|}
            \hline
            Target:&1&2&3&4&5\\
            \hline
            $U_a^c(i)$&$ub$&$lb$&$lb$&$ub$&$lb$\\ \hline
            $U_a^u(i)$&$lb$&$ub$&$lb$&$ub$&$ub$\\ \hline
            
        \end{tabular}
    \end{center}
    corresponding to a type I.A.i equilibrium with $r=s=t=1$ and
    \begin{center}
        \begin{tabular}{|c|c|c|c|c|c|}
            \hline
            Target:&1&2&3&4&5\\
            \hline
            $\alpha_i$&0&1&7/10&1&3/10\\ \hline
             $\beta_i$&0&0&$\approx 0.1509$&1&$\approx0.8491$\\ \hline
        \end{tabular}
    \end{center}
\end{example}

\section{Nearest Additive Games}

The structural characterization presented in the previous sections and the extensive body of literature concerning additive security games suggest the question: how can we utilize the theory of additive security games to address the solution of non-additive security games? One line of inquiry concerns finding, for any non-additive game, an additive game that is `{\it close}' in terms of payoffs. This strategy is adopted to study the class of potential games in \cite{candogan2013near} in which a notion of {\it near-potential games} is developed through introducing a distance measure on the space of all games. Motivated by this work, in this section we propose an answer to the question: Given a non-additive security game, what is the {\it nearest} additive game?

Let $(A,B,k_a,k_d,U_a^c,U_a^u,U_d^c,U_d^u)$ be a non-additive security game over a target set $T=[m]$ expressed as a bimatrix game. In this section, assume that for each $1\leq k\leq m$ we have fixed an ordering of each of the sets $2^{\leq k}=\{S\subseteq [m]:|S|\leq k\}$. Consider that each entry of each payoff matrix is the result of the evaluation of some set function $f:2^{[m]}\to \R$ on a collection of targets. Accordingly, we first define a notion of {\it nearest additive function} to $f$.  

\begin{definition}
    Let $\mathbf{x}\in \R^m$. The additive function defined by $\mathbf{x}$ is the map $h_{\mathbf{x}}:2^{[m]}\to \R$ given by \[h_{\mathbf{x}}(S)=\sum_{i\in S}x_i.\]
\end{definition}
We represent all set functions as vectors of their codomains as follows:
\begin{definition}
    For any $f:2^{[m]}\to \R$ and $1\leq k\leq m$, represent $f$ as a vector $I_k(f)\in\R^{|2^{\leq k}|}$ by
    \[\langle I_k(f)\rangle_i=f(S_i)\] where $S_i$ is the $i^{th}$ element of $2^{\leq k}$.
\end{definition}
We now introduce the notion of {\it nearest additive function}.
\begin{definition}\label{nearestfunctiondef}
    Let $1\leq k\leq m$ and $f:2^{[m]}\to \R$. A $k-$nearest additive function to $f$ is a function $h_{\mathbf{x}}$ such that \[\mathbf{x}=\argmin_{\mathbf{y}}||I_k(f)-I_k(h_{\mathbf{y}})||\]
    where $||\cdot||$ is the euclidean norm on $\R^{|2^{\leq k}|}$. When such a function exists, denote it by $[f]_{k}$.
\end{definition}

\begin{theorem}\label{additiveexistunique}
    For any $1\leq k\leq m$ and $f:2^{[m]}\to \R$, $[f]_{k}$ exists and is unique.
\end{theorem}
\begin{proof}
    Let $V=\{I_k(h_\mathbf{x})\in \R^{|2^{\leq k}|}:\mathbf{x}\in \R^m\}$ be the space of functions that are additive on sets in $2^{\leq k}$. For each $i\in [m]$, let $\chi_i:2^m\to \R$ be defined by \[\chi_i(S)=\begin{cases} 1&\text{ if }i\in S\\0&\text{ otherwise}\end{cases}.\] Note that $V=\text{span}\{I_k(\chi_1),\dots,I_k(\chi_m)\}$ as \[I_k(h_\mathbf{x})=\sum_{i=1}^m x_i I_k(\chi_i).\] Thus $V$ is a linear subspace of $\R^{|2^{\leq k}|}$ and the projection of any $I_k(f)$ onto $V$ is unique and meets the condition of Definition~\ref{nearestfunctiondef}.
\end{proof}
A direct consequence of Theorem~\ref{additiveexistunique} is that we can obtain the following closed-form expression for $\mathbf{x}$:
\begin{corollary}
    For any $1\leq k\leq m$ and $f:2^{[m]}\to \R$, the unique $\mathbf{x}\in \R^m$ meeting the condition of Definition~\ref{nearestfunctiondef} has the form
    \begin{equation}\label{xi}
    x_i=(\sigma_1+\sigma_2)\sum_{i\in S\atop |S|\leq k}f(S)+\sigma_2\sum_{j\neq i}\sum_{j\in S\atop |S|\leq k}f(S)
    \end{equation}
    where \[\sigma_1={1\over \sum_{t=0}^{k-1}{m-1\choose t}-\sum_{t=0}^{k-2}{m-2\choose t}}\]
    and \[\sigma_2={\sum_{t=0}^{k-2}{m-2\choose t}\over \left(\sum_{t=0}^{k-1}{m-1\choose t}-\sum_{t=0}^{k-2}{m-2\choose t}\right)\left(\sum_{t=0}^{k-1}{m-1\choose t}-(m+1)\sum_{t=0}^{k-2}{m-2\choose t}
    \right)}.\]
\end{corollary}
\begin{proof}
    Let $g:\R^m\to \R$ be defined by $g(\mathbf{x})=\|I_k(h_{\mathbf{x}})-I_k(f)\|^2$. Note that
    \begin{equation*}
    \begin{aligned}
    {\partial g\over \partial x_i}&=-2\left(\sum_{i\in S\in 2^{\leq k}}f(S)-\sum_{j=1}^m\left|\{S\in 2^{\leq k}:i,j\in S\}\right|x_j\right)\\
    &=-2\left(\sum_{i\in S\in 2^{\leq k}}f(S)-\sum_{i\neq j}\sum_{t=0}^{k-2}{m-2\choose t}x_j-\sum_{t=0}^{k-1}{m-1\choose t}x_i\right).
    \end{aligned}
    \end{equation*}
    Therefore $\nabla g=0$ if and only if for all $i$,
    \[\sum_{i\in S\in 2^{\leq k}}f(S)=\sum_{i\neq j}\sum_{t=0}^{k-2}{m-2\choose t}x_j+\sum_{t=0}^{k-1}{m-1\choose t}x_i.\]
    Equivalently,
    \begin{equation*}
    \left(\left[\sum_{t=0}^{k-2}{m-2\choose t}\right](\mathbbm{1}_{m\times m}-I_m)+\left[\sum_{t=0}^{k-1}{m-1\choose t}\right]I_m\right)\mathbf{x}=\boldsymbol{\gamma}
    \label{sys}
    \end{equation*}
    where $\gamma_i=\sum_{i\in S\in 2^{\leq k}}f(S)$. Note that for any $a,b\in\R$ we have
    \begin{equation*}\label{inverse_formula}[(a-b)I_m+b\mathbbm{1}_{m\times m}]^{-1}={1\over a-b}I+{b\over (a-b)(a-b(m+1))}\mathbbm{1}_{m\times m}.\end{equation*} Thus, $x_i$ has the specified form.
\end{proof}
We now propose the notion of {\it nearest additive game}.
\begin{definition}
    Let $G=(A,B,k_a,k_d,U_a^c,U_a^u,U_d^c,U_d^u)$ be a (possibly non-additive) security game. The nearest additive game to $G$ is the additive game defined by $k_a,k_d,[U_a^c]_{k_a},[U_a^u]_{k_a},[U_d^c]_{k_a},[U_d^u]_{k_a}$.
\end{definition}

\begin{remark}
    Note that we choose the $k_a$-nearest additive payoff functions in our definition of nearest additive game. We have, for $f:2^{[m]}\to \R$ and $k\leq k^\prime$ \[\min_{\mathbf{x}}||I_k(f)-I_k(h_{\mathbf{x}})||\leq\min_{\mathbf{x}}||I_{k^\prime}(f)-I_{k^\prime}(h_{\mathbf{x}})||\]
    as \[\|I_{k^\prime}(h_{\mathbf{x}^\prime})-I_{k^\prime}(f)\|^2\geq \sum_{S\in 2^{\leq k}}(h_{\mathbf{x}^\prime}(S)-f(S))^2\geq \|I_k(h_{\mathbf{x}})-I_k(f)\|^2. \]
    Thus, the approximation is better for smaller values of $k$. The payoff functions will be evaluated on sets of size at most $k_a$, and so we select this value of $k$ in our nearest additive game.
\end{remark}

\begin{remark}
    The nearest additive game to an additive game is the additive game itself. If $f=h_{\mathbf{x}}$ is an additive function, $\mathbf{y}=\mathbf{x} $ clearly minimizes $\|I_k(f)-I_k(h_{\mathbf{y}})\|=0$.
\end{remark}

Suppose $(A,B)$ and $(\overline{A},\overline{B})$ are a non-additive security game and its nearest additive game respectively. Let $(p,q)$ and $(\overline{p},\overline{q})$ be Nash equilibrium profiles in the original and nearest additive games respectively. Observe that $p^tBq$ is the actual expected outcome to the defender in the original game and $p^tB\overline{q}\leq p^tBq$ is the expected outcome to the defender when she plays the original game using his strategy from the nearest additive game. The relative error (for the defender) in using the approximation, therefore, is given by $|p^tBq-p^tB\overline{q}|/p^tBq$.

\begin{figure}
\caption{The underlying graph of our 14 bus power system. In the zero-sum game we take $k_a=k_d=2$. The values of $(\alpha_i,\beta_i)$ for each edge $i$ are shown for the nearest additive game (left) and the additive game defined by the 1-line disturbance values (right). Power supply and demand values for each node correspond to those in the standard IEEE 14-Bus power system test case, and line reactances are given in the table found in Appendix C.}
\begin{multicols}{2}
    \includegraphics[scale=0.28]{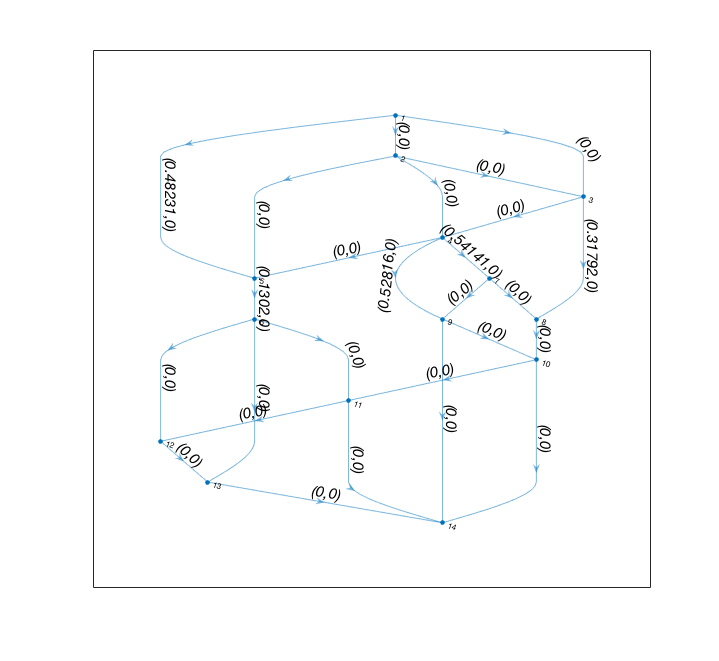}\vfill\null
    \columnbreak
    \includegraphics[scale=0.28]{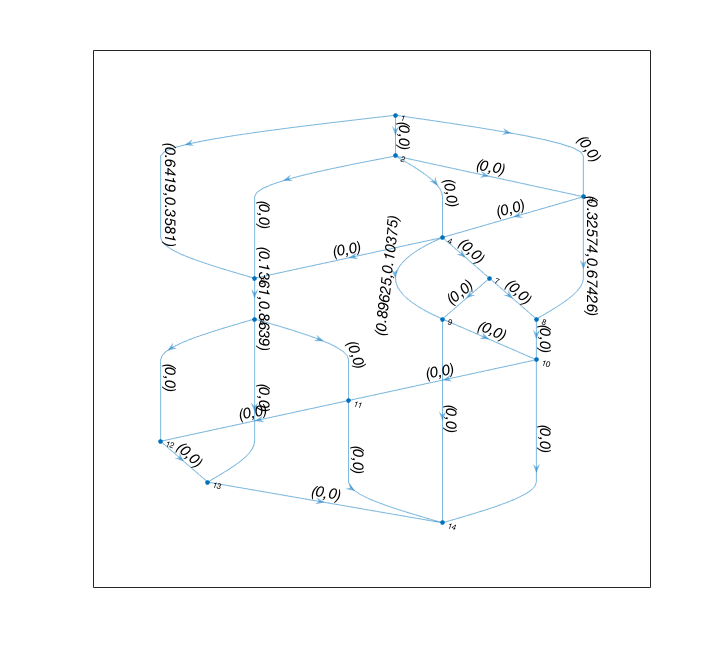}
\end{multicols}
\vspace{-2em}

\label{powergrid}
\end{figure}
\begin{example}
An example of a function which is well-approximated by an additive function is given in \cite{soltan2017analyzing} wherein the notion of {\it disturbance value} is introduced to quantify the effect of $k$-line failures in power grids. The $k-$line disturbance value of a set of line failures of size $k$ is well approximated by the sum of the $k$ individual 1-line disturbance values in large power grids. That is, $\delta_k(\{e_1,\dots,e_k\})\approx \sum_{i=1}^k\delta_1(\{e_i\})$. As in \cite{emadi2021game}, we define a zero-sum security game over the edges of a power grid in which resources are fully protective and the payoffs to the attacker for a set of successfully attacked lines is given by the disturbance value. The relative error in $v^*$ when the nearest 2-additive game is solved is less than half of the relative error when the 1-line disturbance values are taken to be the additive constants. In the game defined on the grid in Figure~\ref{powergrid} we have $v^*_{\text{actual}}\approx 15.123$, $v^*_{\text{addapprox}}\approx 17.815$, and $v^*_{\text{additive}}\approx 20.614$ giving $error_{addapprox}\approx 17.8\%$ and $error_{additive}\approx 36.3\%$.
\label{disturbancevalueexample}
\end{example}

The above examples suggest further study to address the following (open) problems:

\begin{problem}
    For what classes of $U_a^c,U_a^u,U_d^c,U_d^u$ can $|v_d^*-[v_d^*]_k|$, where $[v_d^*]$ is the defender expected outcome in the nearest additive game, be bounded?
\end{problem}
\begin{problem}
    Let $(p^*,q^*),(\overline{p}^*,\overline{q}^*)$ denote Nash equilibrium profiles in a non-additive security game and the nearest additive game respectively. Let $(A,B)$ and $(\overline{A},\overline{B})$ be the bimatrix representations of these games. For what classes of $U_a^c,U_a^u,U_d^c,U_d^u$ can $|p^*Bq^*-p^*B\overline{q}^*|$ be bounded? That is: when is the defender justified in using the strategy she computes using the nearest additive game approximation?
\end{problem}

\section{Conclusion}

In this work, we have contributed to the theory of Nash equilibria in additive security games with singleton schedules and multiple homogeneous attacker and defender resources by proposing a novel structural approach. Our structural analysis of the necessary properties of equilibria in security games have lead to a characterization of the possible types of Nash equilibria in these games and consequently, a new algorithm for computation of such equilibria. Our characterization yields closed-form expressions for the player expected outcomes at equilibrium as well as explicit feasibility conditions for equilibria of each type. We have shown that in the case of fully protective resources, a restricted structure of equilibrium types emerges and that our solution algorithm can be made more efficient. Finally, to address the problem of maximizing the defender expected utility at equilibrium given that the attacker payoffs may be perturbed, we show that (fundamentally owing to the NP-hardness of computing Stackelberg equilibria in the case of multiple attacker resources) this problem is NP hard for Stackelberg equilibria and propose a pseudopolynomial time algorithm for the case of Nash equilibria under a disjointness assumption.

We propose the following as promising directions for future work:
    \begin{itemize}
        \item Extension of our structural approach to the classes of games \cite{korzhyk2010complexity,letchford2013solving} with schedule structures admitting polynomial-time solution algorithms and to repeated games.
        \item Determination of the computational complexity of Problem~\ref{optimizationproblem} and development approximation techniques or approaches that solve typical instances efficiently for the case of Stackelberg equilibria. 
        \item Utilization of the comprehensively understood class of additive security games to compute approximate solutions to non-additive security games (which are NP-hard in general) through the development of projection frameworks. In particular, we seek a classification of the {\it nearly additive} games.
    \end{itemize}


\pagebreak
\appendix

\section{Closed-Form Expressions for Player Expected Outcome at Equilibrium}

        \noindent{\bf Type I.A.i:}
        \begin{equation}
        \begin{aligned}
        v_a^*&=\sum_{i\in I_3}U_a^u(i)+\sum_{i\in I_9}U_a^c(i)+\sum_{i\in I_5}{(k_a-s-t)\left(k_d-t-\sum_{j\in I_5}{U_a^u(j)\over \Delta_a(j)}\right)\over\Delta_d(i)\left(\sum_{j\in I_5}{1\over\Delta_d(j)}\right)\left(\sum_{j\in I_5}{1\over\Delta_a(j)}\right)}\\
        v_d^*&=\sum_{i\in I_3}U_d^u(i)+\sum_{i\in I_9}U_d^c(i)+\sum_{i\in I_5}{(t+s-k_a)\left(k_d-t-\sum_{j\in I_5}{U_a^u(j)\over \Delta_a(j)}\right)\over\Delta_a(i)\left(\sum_{j\in I_5}{1\over \Delta_d(j)}\right)\left(\sum_{j\in I_5}{1\over \Delta_a(j)}\right)}
        \end{aligned}
        \end{equation}

        \noindent{\bf Type I.A.ii:}
        \begin{equation}
        \begin{aligned}
            v_a^*&=\alpha_{j^{[2]}}U_a^u(j^{[2]})+\sum_{i\in I_3}U_a^u(i)+\sum_{i\in I_9}U_a^c(i)\\
            &+\sum_{i\in I_5}{(k_a{-}s{-}t{-}\alpha_{j^{[2]}})[(U_a^u(i){-}U_a^u(j^{[2]}))U_a^c(i){+}(\Delta_a(i){-}U_a^u(i){+}U_a^u(j^{[2]}))U_a^u(i)]\over \Delta_a(i)\Delta_d(i)\sum_{j\in I_5}{1\over \Delta_d(j)}}\\
            v_d^*&=\alpha_{j^{[2]}}U_d^u(j^{[2]})+\sum_{i\in I_3}U_d^u(i)+\sum_{i\in I_9}U_d^c(i)\\
            &+\sum_{i\in I_5}{(k_a{-}s{-}t{-}\alpha_{j^{[2]}})[(U_a^u(i){-}U_a^u(j^{[2]}))U_d^c(i){+}(\Delta_a(i){-}U_a^u(i){+}U_a^u(j^{[2]}))U_d^u(i)]\over\Delta_a(i)\Delta_d(i)\sum_{j\in I_5}{1\over\Delta_d(j)}}
        \end{aligned}
        \end{equation}

        \noindent{\bf Type I.A.iii:}
        \begin{equation}
        \begin{aligned}
        v_a^*&=\alpha_{j^{[8]}}U_a^c(j^{[8]})+\sum_{i\in I_3}U_a^u(i)+\sum_{i\in I_9}U_a^c(i)\\
        &+\sum_{i\in I_5}{(k_a-s-t-\alpha_{j^{[8]}})[(U_a^u(i)-U_a^c(j^{[8]}))U_a^c(i)+(\Delta_a(i)-U_a^u(i)+U_a^c(j^{[8]}))U_a^u(i)]\over\Delta_a(i)\Delta_d(i)\sum_{j\in I_5}{1\over\Delta_d(j)}}\\
        v_d^*&=\alpha_{j^{[8]}}U_d^c(j^{[8]})+\sum_{i\in I_3}U_d^u(i)+\sum_{i\in I_9}U_d^c(i)\\
        &+\sum_{i\in I_5}{(k_a-s-t-\alpha_{j^{[8]}})[(U_a^u(i)-U_a^c(j^{[8]}))U_d^c(i)+(\Delta_a(i)-U_a^u(i)+U_a^c(j^{[8]}))U_d^u(i)]\over \Delta_a(i)\Delta_d(i)\sum_{j\in I_5}{1\over \Delta_d(j)}}
        \end{aligned}
        \end{equation}

        \noindent{\bf Type I.B.i:}
        \begin{equation}
        \begin{aligned}
        v_a^*&=U_a^u(j^{[6]})+\beta_{j^{[6]}}\Delta_a(j^{[6]})+\sum_{i\in I_3}U_a^u(i)+\sum_{i\in I_9}U_a^c(i)\\
        &+\sum_{i\in I_5}{\Delta_d(j^{[6]})\over \Delta_d(i)\sum_{j\in I_5}{1\over \Delta_d(j)}}\left[\sum_{j\in I_5}{U_a^u(j)\over \Delta_a(j)}{+}{1\over \Delta_a(i)}[(U_a^u(i){+}U_a^c)(t{+}\beta_{j^{[6]}}){-}\Delta_a(i)k_d]\right]\\
        v_d^*&=U_d^u(j^{[6]})+\beta_{j^{[6]}}\Delta_d(j^{[6]})+\sum_{i\in I_3}U_d^u(i)+\sum_{i\in I_9}U_d^c(i)\\
        &+\sum_{i\in I_5}{\Delta_d(j^{[6]})\over \Delta_a(i)\Delta_d(i)\sum_{j\in I_5}{1\over \Delta_a(j)}}\left[\Delta_d(i)k_d{+}(U_d^u(i){+}U_d^c(i))(t{+}\beta_{j^{[6]}}){-}\Delta_d(i)\sum_{j\in I_5}{U_u^a(j)\over \Delta_a(j)}\right]
        \end{aligned}
        \end{equation}

        \noindent{\bf Type I.B.ii:}
        \begin{equation}
        \begin{aligned}
        v_a^*&=\sum_{i\in I_3}U_a^u(i)+\sum_{i\in I_9}U_a^c(i)+\left[k_a-s-t-1-\sum_{i\in I_5}{\Delta_d(j^{[6]})\over \Delta_d(i)}\right]U_a^u(j^{[2]})+U_a^u(j^{[6]})\\
        &-\left[k_d-t-\sum_{i\in I_5}{U_a^u(i)-U_a^u(j^{[2]})\over \Delta_a(i)}\right]U_a^u(j^{[6]})\\
        &+\sum_{i\in I_5}{\Delta_d(j^{[6]})[(U_a^u(i)-U_a^u(j^{[2]}))U_a^c(i)+(U_a^u(j^{[2]})-U_a^c(i))U_a^u(i)]\over \Delta_a(i)\Delta_d(i)}\\
        v_d^*&=\sum_{i\in I_3}U_d^u(i)+\sum_{i\in I_9}U_d^c(i)+\left[k_a-s-t-1-\sum_{i\in I_5}{\Delta_d(j^{[6]})\over \Delta_d(i)}\right]U_d^u(j^{[2]})+U_d^u(j^{[6]})\\
        &-\left[k_d-t-\sum_{i\in I_5}{U_a^u(i)-U_a^u(j^{[2]})\over \Delta_a(i)}\right]U_d^u(j^{[6]})\\
        &+\sum_{i\in I_5}{\Delta_d(j^{[6]})[(U-a^u(i)-U_a^u(j^{[2]}))U_d^c(i)+(U_a^u(j^{[2]})-U_a^c(i))U_d^u(i)]\over \Delta_a(i)\Delta_d(i)}
        \end{aligned}
        \end{equation}
        \vfill\null
        
        \noindent{\bf Type I.B.iii:}
        \begin{equation}
        \begin{aligned}
        v_a^*&=\sum_{i\in I_3}U_a^u(i)+\sum_{i\in I_9}U_a^c(i)+U_a^u(j^{[6]})-\left[k_a-t-1-\sum_{i\in I_5}{U_a^u(i)-U_a^c(j^{[8]})\over \Delta_a(i)}\right]\Delta_a(j^{[6]})\\
        &+\left[k_a-s-t-\sum_{i\in I_5}{\Delta_d(j^{[6]})\over\Delta_d(i)}\right]U_a^c(j^{[8]}))\\
        &+\sum_{i\in I_5}{\Delta_d(j^{[6]})[(U_a^u(i)-U_a^c(j^{[8]})U_a^c(i)+(U_a^c(j^{[8]}-U_a^c(i))U_a^u(i)]\over\Delta_a(i)\Delta_d(i)}\\
        v_d^*&=\sum_{i\in I_3}U_d^u(i)+\sum_{i\in I_9}U_d^c(i)+U_d^u(j^{[6]})+\left[k_a-t-1-\sum_{i\in I_5}{U_a^u(i)-Ua^c(j^{[8]})\over\Delta_a(i)}\right]\Delta_d(j^{[6]})\\
        &+\left[k_a-s-t-\sum_{i\in I_5}{\Delta_d(j^{[6]})\over \Delta_d(i)}\right]U_d^c(j^{[8]})\\
        &+\sum_{i\in I_5}{\Delta_d(j^{[6]})[(U_a^u(i)-U_a^c(j^{[8]})U_d^c(i)+(U_a^c(j^{[8]}-U_a^c(i))U_d^u(i)]\over\Delta_a(i)\Delta_d(i)}
        \end{aligned}
        \end{equation}

        \noindent{\bf Type II:}
        \begin{equation}v_a^*=\sum_{i\in I_9}U_a^c(i)\hspace{2em}v_d^*=\sum_{i\in I_9}U_d^c(i)\end{equation}
	\vfill\null

\section{Closed-Form Expressions for Player Expected Outcome at Equilibrium: Fully Protective Resources}

        \noindent{\bf I.A.i with $I_9=\emptyset$
        \begin{equation}
        \begin{aligned}
            v_a^*&=\sum_{i\in I_3}U_a^u(i)+\sum_{i\in I_5}{k_a-s\over U_d^u(i)\sum_{j\in I_5}{1\over U_d^u(j)}}\left[U_a^u(i)-{k_a-m+r+s\over \sum_{j\in I_5}{1\over U_a^u(j)}}\right]\\
            v_d^*&=\sum_{i\in I_3}U_d^u(i)+\sum_{i\in I_5}{k_a-s\over U_a^u(i)\sum_{j\in I_5}{1\over U_d^u(j)}}\left[U_a^u(i)-{k_a-m+r+s\over \sum_{j\in I_5}{1\over U_a^u(j)}}\right]
        \end{aligned}
        \end{equation}
        
        \noindent{\bf I.A.ii with $I_9=\emptyset$
        \begin{equation}
        \begin{aligned}
            v_a^*&=\alpha_{j^{[2]}}U_a^u(j^{[2]})+\sum_{i\in I_3}U_a^u(i)+\sum_{i\in I_5}{(k_a-s-\alpha_{j^{[2]}})(U_a^u(i)-U_a^u(j^{[2]}))\over U_d^u(i)\sum_{j\in I_5}{1\over U_d^u(j)}}\\
            v_d^*&=\alpha_{j^{[2]}}U_d^u(j^{[2]})+\sum_{i\in I_3}U_d^u(i)+\sum_{i\in I_5}{(k_a-s-\alpha_{j^{[2]}})(U_a^u(i)-U_a^u(j^{[2]}))\over U_a^u(i)\sum_{j\in I_5}{1\over U_d^u(j)}}
        \end{aligned}
        \end{equation}
        
        \noindent{\bf I.B.i with $I_9=\emptyset$
        \begin{equation}
        \begin{aligned}
            v_a^*&=\sum_{i\in I_3}U_a^u(i)+U_a^u(j^{[6]})(1-\beta_{j^{[6]}})+\sum_{i\in I_5}{U_d^u(j^{[6]})\over U_d^u(i)}\left[U_a^u(i){-}{k_d{-}m{+}r{+}s{+}1{-}\beta_{j^{[6]}}\over \sum_{j\in I_5}{1\over U_a^u(j)}}\right]\\
            v_d^*&=\sum_{i\in I_3}U_d^u(i)+U_a^u(j^{[6]})(1-\beta_{j^{[6]}})+\sum_{i\in I_5}{U_d^u(j^{[6]})\over U_a^u(i)}\left[U_a^u(i){-}{k_d{-}m{+}r{+}s{+}1{-}\beta_{j^{[6]}}\over \sum_{j\in I_5}{1\over U_a^u(j)}}\right]
        \end{aligned}
        \end{equation}
        
        \noindent{\bf I.B.ii with $I_9=\emptyset$
        \begin{equation}
        \begin{aligned}
            v_a^*&=\left[k_a{-}s{-}1{-}\sum_{i\in I_5}{U_d^u(j^{[6]})\over U_d^u(i)}\right]U_a^u(j^{[2]})+\sum_{i\in I_3}U_a^u(i)+\sum_{i\in I_3}U_a^u(i)\\
            &+\sum_{i\in I_5}{U_d^u(j^{[6]})[U_a^u(i)-U_a^u(j^{[2]})\over U_d^u(i)]}\\
            v_d^*&=\left[k_a{-}s{-}1{-}\sum_{i\in I_5}{U_d^u(j^{[6]})\over U_d^u(i)}\right]U_d^u(j^{[2]})+\sum_{i\in I_3}U_d^u(i)+\sum_{i\in I_3}U_d^u(i)\\
            &+\sum_{i\in I_5}{U_d^u(j^{[6]})[U_a^u(i)-U_a^u(j^{[2]})\over U_a^u(i)]}
        \end{aligned}
        \end{equation}
        
        \noindent{\bf I.B.iii with $I_1\cup I_5=\emptyset$
      	\begin{equation}v_a^*=\sum_{i\in I_3}U_a^u(i)\hspace{2em}v_a^*=\sum_{i\in I_3}U_d^u(i)\end{equation}

\section{Line Reactance Values for Example~\ref{disturbancevalueexample}}

\begin{center}
\begin{tabular}{|c|c||c|c|}
        \hline
        Line&Reactance&Line&Reactance\\
        \hline
        1,2	&0.01938&6,12&0.022\\\hline
        1,3	&0.05403&6,13&0.0137\\\hline
        1,5	&0.04699&7,8&0.03181\\\hline
        2,3	&0.05811&7,9&0.12711\\\hline
        2,4	&0.05695&8,10&0.08205\\\hline
        2,5	&0.06701&9,10&0.22092\\\hline
        3,4	&0.01335&9,14&0.17093\\\hline
        3,8	&0.061&10,11&0.34\\\hline
        4,5	&0.086&10,14&0.19\\\hline
        4,7	&0.154&11,12&0.27\\\hline
        4,9	&0.09498&11,14&0.085\\\hline
        5,6	&0.12291&12,13&0.34\\\hline
        6,11&0.06615&13,14&0.345\\\hline
\end{tabular}
\end{center}

\end{document}